\documentclass[a4paper,12pt]{article}
\usepackage[utf8]{inputenc}
\usepackage[T1]{fontenc}
\usepackage[left=2.5cm,right=2.5cm,top=3cm,bottom=3.5cm]{geometry}
\usepackage{eucal}
\usepackage{microtype}
\usepackage{csquotes}
\usepackage{amssymb,amsfonts,stmaryrd,amsthm,mathtools}
\usepackage{float}
\usepackage{mathrsfs}  
\usepackage{tikz}
\usetikzlibrary{cd}
\usepackage{enumitem}
\usepackage[normalem]{ulem}
\setlist[enumerate,1]{label={\roman*)}}
\usepackage{graphicx,subcaption}
\graphicspath{{figures/}}
\usepackage{tikz}
\usetikzlibrary{patterns,patterns.meta}
\usetikzlibrary{calc}

\newcommand{\X}{\mathfrak{X}}
\newcommand{\sode}{\Gamma}

\DeclarePairedDelimiter{\gen}{\langle}{\rangle}
\newcommand{\Cinfty}{\mathscr{C}^\infty}

\newcommand{\parent}[1]{
\left( #1 \right)
}
\newcommand{\pder}[2]{ \frac{\partial #1}{\partial #2} }

\theoremstyle{plain}
\newtheorem{theorem}{Theorem}
\newtheorem*{theorem*}{Theorem}
\newtheorem{lemma}[theorem]{Lemma}
\newtheorem*{lemma*}{Lemma}
\newtheorem{proposition}[theorem]{Proposition}
\newtheorem*{proposition*}{Proposition}
\newtheorem{corollary}[theorem]{Corollary}  
\newtheorem*{corollary*}{Corollary}
\theoremstyle{definition}
\newtheorem{definition}{Definition}
\newtheorem*{definition*}{Definition}
\newtheorem{example}{Example}
\newtheorem*{example*}{Example}
\newtheorem{remarkx}{Remark}
\newtheorem*{remarkx*}{Remark}
\newenvironment{remark}
  {\pushQED{\qed}\remarkx}
  {\popQED\endremarkx}

\newtheorem*{conjecture*}{Conjecture}

\newtheorem*{problem*}{Problem}

\newcommand*{\NN}{\mathbb{N}}

\newcommand*{\RR}{\mathbb{R}}
\let\R\RR

\newcommand{\mommap}{\mathbf{J}}

\newcommand{\Fext}{F^{\mathrm{ext}}}
\newcommand{\Ftext}{\tilde{F}^{\mathrm{ext}}}
\newcommand{\Ffr}{F^{\mathrm{fr}}}
\newcommand{\Ftfr}{\tilde{F}^{\mathrm{fr}}}
\newcommand{\fd}{f_d}
\newcommand{\fdp}{f^{+}_d}
\newcommand{\fdm}{f^{-}_d}
\newcommand{\fdpm}{f^{\pm}_d}
\newcommand{\ffrd}{f^{\mathrm{fr}}_d}
\newcommand{\ffrdp}{f^{\mathrm{fr,+}}_d}
\newcommand{\ffrdm}{f^{\mathrm{fr,-}}_d}
\newcommand{\ffrdpm}{f^{\mathrm{fr,\pm}}_d}
\newcommand{\fextd}{f^{\mathrm{ext}}_d}
\newcommand{\fextdp}{f^{\mathrm{ext,+}}_d}
\newcommand{\fextdm}{f^{\mathrm{ext,-}}_d}
\newcommand{\fextdpm}{f^{\mathrm{ext,\pm}}_d}

\newcommand{\Sendo}{\mathcal{S}} 
\newcommand{\Sendoadj}{\Sendo^\ast} 
\newcommand{\Sendobar}{\overline{\Sendo}}

\newcommand\restr[2]{{
  \left.\kern-\nulldelimiterspace 
  #1 
  \right|_{#2} 
}}

\newcommand{\T}{{\mathsf T}}
\newcommand{\cT}{\T^{\ast}}
\newcommand*{\dd}{\mathrm{d}}
\newcommand*{\DD}{{\mathsf D}}

\DeclareMathOperator{\Leg}{Leg}
\newcommand*{\contr}[1]{\iota_{#1}}
\newcommand*{\liedv}[1]{\mathcal{L}_{#1}}
\newcommand*{\Reeb}{\mathcal{R}} 
\newcommand*{\evol}{\mathcal{E}} 

\let\coloneqq\defeq

\let\oldemph\emph
\let\emph\textbf

\usepackage[style=ext-numeric,citestyle=numeric-comp,sorting=nyt,sortcites, backend=biber, giveninits=true, maxnames=99]{biblatex}
\bibliography{biblio}
\DeclareFieldFormat[article,periodical]{volume}{\mkbibbold{#1}}
\DeclareFieldFormat[article,periodical]{number}{\mkbibparens{#1}}
\renewbibmacro{in:}{%
  \ifentrytype{article}
    {}
    {\printtext{\bibstring{in}\intitlepunct}}}

\DeclareFieldFormat{issuedate}{#1}
\renewcommand{\jourvoldelim}{\addcomma\space}

\renewbibmacro*{journal+issuetitle}{%
  \usebibmacro{journal}%
  \setunit*{\jourvoldelim}%
  \iffieldundef{series}
    {}
    {\setunit*{\jourserdelim}%
     \printfield{series}%
     \setunit{\servoldelim}}%
  \usebibmacro{volume+number+eid}%
  \setunit{\addcolon\space}%
  \usebibmacro{issue}%
  \setunit{\bibpagespunct}%
  \printfield{pages}%
  \setunit{\space}%
  \usebibmacro{issue+date}%
  \newunit}

\renewbibmacro*{note+pages}{%
  \printfield{note}%
  \newunit}

\DeclareFieldFormat[article,periodical]{date}{\mkbibparens{#1}}

\AtEveryBibitem{\clearfield{month}}
\AtEveryBibitem{\clearfield{day}}

\AtEveryBibitem{
  \ifentrytype{online}{}{
    \ifentrytype{thesis}{}{
      \ifentrytype{unpublished}{}{ 
        \clearfield{url}
      }
    }
  }
  \ifentrytype{unpublished}{}{ 
  \clearfield{note}
  }
  \clearfield{language}

  \ifentrytype{book}{\clearfield{pages}}{}
}

\DeclareSourcemap{
  \maps[datatype=bibtex]{
    \map[overwrite=true]{
      \step[fieldset=language, null]
      \step[fieldset=address, null]
     \step[fieldset=pagetotal, null]
    }
  }
}

\DeclareSourcemap{
    \maps[datatype=bibtex]{
      \map{
        \step[fieldsource=eprint,final]
        \step[fieldset=url,null]
        \step[fieldset=doi,null]
      }  
    }
  }

\AtEveryBibitem{\clearfield{pubstate}} 
\AtEveryBibitem{\clearfield{version}} 

\usepackage[unicode=true, pdfusetitle, colorlinks=true, allcolors=blue, bookmarks=false]{hyperref}
\usepackage[hyphenbreaks]{breakurl}
\hypersetup{breaklinks=true}

\begin{document}

\title{\sffamily Geometric integrators for adiabatically\\ closed simple thermodynamic systems} 

\author{\sffamily Jaime Bajo$^{1}$, Manuel de Le\'on$^{2,3}$, 
\sffamily and Asier L\'opez-Gord\'on$^{4}$\footnote{Author to whom correspondence should be addressed: \href{mailto:alopez-gordon@impan.pl}{alopez-gordon@impan.pl}.}
\\
\small
$^{1}$ETH Zürich, Switzerland\\
\small
$^{2}$Institute of Mathematical Sciences, Spanish National Research Council, Madrid, Spain\\ 
\small
$^{3}$Royal Academy of Exact, Physical and Natural Sciences of Spain, Madrid, Spain\\
\small
$^{4}$Institute of Mathematics, Polish Academy of Sciences, Warsaw, Poland
}

\date{\vspace*{-1cm}}

\maketitle

\begin{abstract}
\noindent 
A variational formulation for non-equilibrium thermodynamics was developed by Gay-Balmaz and Yoshimura. In a recent article, the first two authors of the present paper introduced partially cosymplectic structures as a geometric framework for thermodynamic systems, recovering the evolution equations obtained variationally. In this paper, we develop a discrete variational principle for adiabatically closed simple thermodynamic systems, which can be utilised to construct numerical integrators for the dynamics of such systems. The effectiveness of our method is illustrated with several examples. 

\medskip

\noindent
\textbf{Keywords:} thermodynamic systems, cosymplectic structures, discrete thermodynamic Euler--Lagrange equations, geometric integrators 


\noindent
\textbf{MSC 2020 classes:} 53Z30, 80A05, 53D15 	
\end{abstract}

\tableofcontents


\section{Introduction}

The use of differential geometric methods in the study of thermodynamics dates back considerably. Traditionally, the geometry of equilibrium thermodynamics has been mainly studied via contact geometry. In this setting, thermodynamic properties are encoded by Legendre submanifolds of the thermodynamic phase space (see \cite{bravetti0,bravetti2019,lacomba1,lacomba2,mrugala1,mrugala2} and references therein). In \cite{A.d.L+2020}, a similar approach was considered, but using the so-called evolution vector field instead of the contact Hamiltonian vector field.

Recently, Gay-Balmaz and Yoshimura \cite{G.Y2018,Y.G2022} introduced variational principles for the description of thermodynamic systems. Their formulation extends the Hamilton principle of classical mechanics to include irreversible processes by introducing additional phenomenological and variational constraints. In a recent paper \cite{d.B2025}, the first two authors of the present article introduced almost cosymplectic structures, new geometric structures that lead to the same dynamics as Gay-Balmaz and Yoshimura's approach. An almost cosymplectic structure on a $(2n+1)$-dimensional manifold $M$ is a pair $(\omega,\eta)$ consisting of a $2$-form $\omega$ and a $1$-form $\eta$ such that $\omega^n \wedge \eta$ is a volume form on $M$. These structures were already considered in \cite{BracketsDuales}, albeit in a different context. They are just a slight generalisation of cosymplectic structures, which were introduced by Libermann \cite{libermann}. A cosymplectic structure is simply an almost cosymplectic structure $(\omega, \eta)$ such that both $\omega$ and $\eta$ are closed forms. We shall be considering the so-called partially cosymplectic structures, for which $\omega$ is assumed to be closed, but $\eta$ is not. Note that any contact form $\eta$ defines a partially cosymplectic structure $(\dd \eta, \eta)$. 

In the current paper, we take advantage of this geometric description to obtain a discrete description and construct geometric integrators appropriate for treating these thermodynamic systems. Due to the complexity of the different types of systems, we focus on adiabatically closed simple thermodynamic systems. We hope to develop similar procedures for more complex systems, and thus cover the entire spectrum of cases, in a forthcoming publication. Our discrete model is inspired by variational integrators (see the quintessential review by Marsden and West \cite{M.W2001}), as well as their extensions to contact Lagrangian systems \cite{A.M.L+2021,V.B.S2019}.

In \cite{G.Y2018a}, Gay-Balmaz and Yoshimura also developed variational integrators for the nonequilibrium thermodynamics of simple closed systems, following a similar approach to ours. Let us not forget that this type of discretization of systems described by a Lagrangian has been studied for more than five decades (see \cite{M.W2001} and references therein). The main difference between our development and that of Gay-Balmaz and Yoshimura is not only that they consider two entropies (there are no major differences in this regard), but also that our approach is based on the geometric description of the thermodynamic system, which allows us to use the standard techniques of geometric mechanics. Here we have addressed only the case of infinitesimal symmetries and Noether's theorem, as well as the existence of a momentum mapping, but in subsequent papers we propose to address aspects such as coisotropic reduction, Hamilton--Jacobi theory, and the extension to more complex systems. These objectives require a thorough understanding of the underlying geometry, and that is the first step we are currently developing. A more explicit comparison between our approach and that of Gay-Balmaz and Yoshimura can be found in Table~\ref{table:comparison_approaches}.

It is also worth mentioning the numerical methods for a supercritical van der Waals fluid enclosed in a cylinder developed by Fülöp, Szücz and Takács in \cite{Fulop2025}. Their approach is based on the GENERIC framework, {which, as its name suggests, is a very general framework for describing systems with dissipation (see~\cite{P.K.G2018}).}  
{More precisely, the authors consider
a dissipative dynamical system of the form} 
\begin{equation}\label{eq:GENERIC}
    \frac{\dd x^i}{\dd t} = \sum_j \left(L^i_j \frac{\partial E}{\partial x^j} + M^i_j \frac{\partial S}{\partial x^j}\right)\, ,
\end{equation}
where $(L^i_j)$ is a skew-symmetric matrix (representing a presymplectic form), $(M^i_j)$ is a symmetric matrix, while $E$ and $S$ are functions that can be regarded as the total energy and the total entropy of the system, respectively. The authors then consider a generalisation of the symplectic Euler scheme, to discretise the evolution equations of the fluid. They obtain meritorious numerical results. However, they do not develop a general method for discretising thermodynamic systems, but rather construct an ad hoc generalisation of the symplectic method that serves their system under consideration. Moreover, their approach does not consider the geometry of the thermodynamic phase space whatsoever. We would also like to mention that the dynamics \eqref{eq:GENERIC} corresponds to a Rayleigh-type dissipation, i.e. a dissipation proportional to velocity or momentum. Indeed, the symmetric matrix $(M^i_j)$ determines a Rayleigh tensor. We previously studied systems with Rayleigh dissipation in \cite{d.L.L2021,d.L.L2022} and their discretisation in \cite{d.L.L2022a}.

The remainder of the paper is structured as follows. In Section~\ref{sec:preliminaries}, we present partially cosymplectic structures, recall the continuous model for an adiabatically closed simple thermodynamic system and study some results concerning infinitesimal symmetries in this setting. Section~\ref{sec:discrete_model} presents our discrete model for such systems. In Subsection~\ref{sec:discrete_variational_principle}, we develop a discrete variational principle for such systems, and obtain the discrete thermodynamic Euler--Lagrange equations. The discrete flow is analysed in Subsection~\ref{sec:discrete_flows}. A discrete Noether's theorem is obtained in Subsection~\ref{sec:discrete_Noether}. Some examples and simulations are presented in Section~\ref{sec:examples}.

\textbf{Notation and conventions.}
Throughout the paper, all manifolds are assumed to be smooth, connected and second-countable. Maps are assumed to be smooth. Summation over repeated indices is understood. Given a differential $k$-form $\omega$ on a manifold $M$, the kernel of the vector bundle morphism $v \in \T M \mapsto \contr{v} \omega \in \bigwedge^{k-1} (\cT M)$ will be called the kernel of $\omega$ and denoted by $\ker \omega$.

\section{Preliminaries}\label{sec:preliminaries}

In this section, we will present the notion of partially cosymplectic structures introduced in \cite{BracketsDuales,d.B2025}, and recall the continuous model for an adiabatically closed simple thermodynamic system (see \cite{G.Y2018,d.B2025}).

\subsection{Partially cosymplectic structures} \label{subsec:partially_cosymplectic}


\begin{proposition}\label{proposition:equivalent_defs_almost_cosymp}
    Let $M$ be a $(2n+1)$-dimensional manifold equipped with a $2$-form $\omega$ and a $1$-form $\eta$. Then, the following statements are equivalent:
    \begin{enumerate}
        \item $\omega^n\wedge \eta$ is a volume form on $M$,
        \item $\ker \eta \cap \ker \omega = \{0_M\}$, where $0_M$ denotes the zero-section of $\T M$,
        \item the map 
        \begin{equation}\label{eq:flat_omega_eta}
        \begin{aligned}
            \flat_{(\omega, \, \eta)}\colon \T M &\to \cT M\\
             v & \mapsto \contr{v} \omega + \eta(v) \eta
        \end{aligned}
        \end{equation}
        is a vector bundle isomorphism,
        \item the cotangent bundle can be decomposed as the Whitney sum
        \begin{equation}\label{eq:decomposition_cotangent}
            \cT M = \flat_{\omega}(\T M) \oplus \gen{\eta}\, ,
        \end{equation}
        where $\flat_{\omega}\colon \T M \to \cT M$ is the vector bundle morphism given by $\flat_{\omega}(v) = \contr{v}\omega$.
    \end{enumerate}
\end{proposition}

\begin{proof}
    Notice that a top-degree form on a manifold is a volume form if and only if its kernel is the zero-section. The equivalence between the first and second statements then follows from 
    {
    the equivalence between the following statements:
    \begin{enumerate}
        \item $\contr{v}\left(\omega^n \wedge \eta\right) =0 \Rightarrow v = 0$,
        \item $\contr{v} \omega = 0 \text{ and } \contr{v} \eta = 0 \Rightarrow v =0$.
    \end{enumerate} 
    }
    The equivalence between the second and third statements is a trivial linear algebraic observation. Finally, the equivalence between the second and fourth statements can be seen by taking the annihilator from both sides of equation~\eqref{eq:decomposition_cotangent}.
\end{proof}

\begin{definition}
    Let $M$ be a $(2n+1)$-dimensional manifold with a $2$-form $\omega$ and a $1$-form $\eta$. The pair $(\omega, \eta)$ is called an \emph{almost cosymplectic structure} if it satisfies any equivalent condition from Proposition~\ref{proposition:equivalent_defs_almost_cosymp}. 
    
    If $\omega$ is closed, then $(\omega, \eta)$ is called a \emph{partially cosymplectic structure}. If additionally $\eta$ is closed, then $(\omega, \eta)$ is called a \emph{cosymplectic structure}. 
    
    The triple $(M, \omega, \eta)$ is called an \emph{almost cosymplectic manifold}, \emph{partially cosymplectic manifold}, or \emph{cosymplectic manifold}, respectively.
\end{definition}

In such a case, the vector bundle isomorphism $\flat_{(\omega, \, \eta)}\colon \T M \to \cT M$ naturally induces a $\Cinfty(M)$-module isomorphism $\flat_{(\omega, \, \eta)} \colon \X(M) \to \Omega^1(M)$.

\begin{definition}\label{def:Reeb_evolution_vector_fields}
      Let $(M, \omega, \eta)$ be an almost cosymplectic manifold. The \emph{Reeb vector field} $\Reeb$ on $M$ is uniquely determined by $\Reeb = \flat_{(\omega, \, \eta)}^{-1} (\eta)$, or equivalently,
    $$\contr{\Reeb} \omega = 0\, , \quad \contr{\Reeb} \eta = 1\, .$$
    For each function $f\in \Cinfty(M)$, the \emph{evolution vector field} $\evol_{f}$ is uniquely determined by 
    $$\evol_{f} = \flat_{(\omega, \, \eta)}^{-1}(\dd f) + \Reeb\, .$$
    In other words,
    $$\contr{\evol_{f}} \omega = \dd f - \Reeb(f) \eta\, , \quad \contr{\evol_{f}} \eta = \Reeb(f) +1 \, .$$
\end{definition}

\begin{remark}
    A $1$-form $\eta$ on a $(2n+1)$-dimensional manifold $M$ is a contact form if and only if the pair $(\dd \eta, \eta)$ is a partially cosymplectic structure on $M$. The Reeb vector field with respect to both structures is the same one; and the isomorphism between $\T M$ and $\cT M$ defined by $\eta$ is precisely $\flat_{(\eta,\, \dd \eta)}$. For more details, see \cite{de2019contact}.
\end{remark}

Let $(\omega, \eta)$ be a partially cosymplectic manifold. Then, $K = \ker \eta$ is a vector subbundle of the tangent bundle $\T M$. Moreover, the restriction of the $2$-form $\omega$ to its fibers makes $K$ a symplectic vector bundle. However, $K$ is not an involutive distribution on $M$. Indeed, for each pair of sections $X$ and $Y$ of $K$, we have
$$\contr{[X,Y]} \eta = \liedv{X} \contr{Y} \eta - \contr{Y} \liedv{X} \eta =  - \contr{Y} \contr{X} \dd \eta \, ,$$
whose right-hand side is non-vanishing in general.

For additional details on almost and partially cosymplectic structures, refer to \cite{BracketsDuales,d.B2025}.

\subsection{The continuous model: adiabatically closed simple thermodynamic systems}\label{subsec:continuous_model}

The phase space of a mechanical Hamiltonian system can usually be regarded as the cotangent bundle $\cT Q$ of its configuration space $Q$. This means that the position of the system is determined by a point $q\in Q$, and its momentum is a covector on the corresponding fiber: $p_q\in \cT_q Q$. As it is well-known, $\cT Q$ is canonically endowed with a symplectic form $\omega_Q$. Indeed, $\omega_{\mathcal Q} = - d \Theta_{Q}$, where $\Theta_{Q} = p_i \, dq^i$ is the Liouville form on $\cT {Q}$ in bundle coordinates $(q^i, p_i)$ on $\cT Q$ associated with some local coordinates $(q^i)$ on $Q$. Then, we have $\omega_Q = \dd q^i \wedge \dd p_i$ (see \cite{godbillon1969geometrie,abraham,de2011methods}).

Let us recall that a simple thermodynamic system is a macroscopic system for which one thermal variable (e.g., the entropy) and a finite set of non-thermal variables are sufficient to entirely describe the state of the system. The system is called adiabatically closed if there is no exchange of heat, nor of matter, with the exterior. Therefore, in our model, the configuration of an adiabatically closed simple thermodynamic system will be determined by a point on the manifold $M = \cT Q \times \RR$. Its canonical bundle coordinates will be denoted by $(q^i, p_i, S)$, and, physically, $S$ will be interpreted as the entropy of the system. The Hamiltonian function will be a function $H\in \Cinfty(M)$.


Given a fiber bundle $\pi\colon E\to B$, a one-form $\alpha\in \Omega^{1}(E)$ on $E$ is called \emph{semibasic} if $\alpha(Z)$ vanishes for every vertical vector field $Z$ on $E$. If $(x^i, y^a)$ are bundle coordinates on $E$ associated with local coordinates $(x^i)$ on $B$, then $\alpha$ reads
\begin{equation}
    \alpha = \alpha_i (x, y)\, \dd x^i\, .
\end{equation} 
In our model, non-conservative forces (i.e., those that cannot be derived from a Hamiltonian function) will be represented by semibasic one-forms on the vector bundle $\pi\colon M = \cT Q \times \RR \to Q$, where the projection is the natural one. More precisely, in bundle coordinates $(q^i, p_i, S)$, the external force $\Fext$ and the friction force $\Ffr$ will have the local expressions
$$\Fext = \Fext_i (q, p, S)\, \dd q^i\, , \quad \Ffr = \Ffr_i (q, p, S)\, \dd q^i\, .$$
From the Hamiltonian function $H$ and the friction force $\Ffr$, we can define the one-form 
$$\eta=-\frac{\partial H}{\partial S} \dd S-\Ffr\, .$$
In addition, pullbacking the canonical symplectic form $\omega_Q$ on $\cT Q$ by the natural projection $\pi_1\colon \cT Q \times \RR \to \cT Q$, we can define the presymplectic form $\omega = \pi_1^\ast \omega_Q$ of corank $1$ on $M$. Similarly, let $\theta=\pi_1^\ast \theta_Q$, with $\theta_Q$ the canonical one-form on $\cT Q$.
In bundle coordinates $(q^i,p_i,S)$, these forms read
\begin{equation}\label{eq:local_forms_Hamiltonian}
    \eta=-\frac{\partial H}{\partial S} \dd S-\Ffr_i(q,p,S)\dd q^i \, , \quad \theta = p_i \dd q^i\, , \quad \omega=\dd q^i\wedge \dd p_i\, .
\end{equation}
Assuming that $\displaystyle{\pder{H}{S}}$ is nowhere-vanishing, then the pair $(\omega,\eta)$ is a partially cosymplectic structure on $M$, regardless of what expression $\Ffr$ takes. As we shall explain below, for the type of Hamiltonians that we will consider, $\displaystyle{\pder{H}{S}}$ will be the temperature of the system. Hence, our assumption simply means that the temperature of the system is not the absolute zero.
It is worth recalling that by the third law of thermodynamics, the absolute zero cannot be reached adiabatically \cite{Kieu2019}. Hereinafter, we shall thus assume that $(M, \omega, \eta)$ is a partially cosymplectic manifold. The tuple $(M, \omega, \eta, H, \Fext, \Ffr)$ will be called the \emph{Hamiltonian thermodynamic system}.

Consider the $\Cinfty(M)$-module isomorphism $\flat_{(\omega,\, \eta)}\colon \X(M) \to \Omega^1(M)$ defined in the previous subsection. In bundle coordinates, it is given by
$$ 
  \flat_{(\omega,\, \eta)}\left(\frac{\partial}{\partial q^i}\right)=\dd p_i-\Ffr_i\eta\, , \quad
  \flat_{(\omega,\, \eta)}\left(\frac{\partial}{\partial p_i}\right)=-\dd q^i \, , \quad
  \flat_{(\omega,\, \eta)}\left(\frac{\partial}{\partial S}\right)=-\frac{\partial H}{\partial S}\eta \, .
$$
The \emph{evolution vector field} $\evol_{H, \Fext}$ of $H$  \emph{subject to external forces} $\Fext$ is defined by
\begin{equation} \label{EvolutionFieldAdiabaticallyClosedSimple}
    \flat_{(\omega,\, \eta)}\left(\evol_{H, \Fext}\right)=\dd H +\eta -\Fext\, .
\end{equation}
In bundle coordinates, it reads
\begin{equation}\label{eq:evol_H_F_coords}
    \evol_{H, \Fext} 
    = \frac{\partial H}{\partial p_i} \pder{}{q^i} 
    +\left(\Ffr_i +\Fext_i  -\frac{\partial H}{\partial q^i}\right) \pder{}{p_i}
    -\frac{1}{\frac{\partial H}{\partial S}}\frac{\partial H}{\partial p_j}\Ffr_j \pder{}{S}\, .
\end{equation}
Therefore, an integral curve $\sigma(t) = (q^i(t),p_i(t),S(t))$ of $\evol_{H, \Fext}$ satisfies the following system of first-order ordinary differential equations:
\begin{equation}\label{CurveAdiabaticallyClosedSimple}
\begin{aligned}
    \frac{\dd q^i}{\dd t} &= \frac{\partial H}{\partial p_i}\, ,\\
    \frac{\dd p_i}{\dd t} &= -\frac{\partial H}{\partial q^i}+\Ffr_i+\Fext_i\, ,\\
    \frac{\dd S}{\dd t} &= -\frac{1}{\frac{\partial H}{\partial S}}\frac{\partial H}{\partial p_j}\Ffr_j\, .
\end{aligned} 
\end{equation}

\bigskip

\noindent Note that the evolution vector field $\evol_{H, \Fext}$ of $H$  subject to external forces $\Fext$ can be written as
$$\evol_{H, \Fext} = \evol_H - Z_{\Fext}\, ,$$
where $\evol_H$ is the evolution vector field of $H$ (see Definition~\ref{def:Reeb_evolution_vector_fields}) and $Z_{\Fext}$ is uniquely determined by $\flat_{(\omega,\, \eta)}(Z_{\Fext}) = \Fext$, i.e.,
$$\contr{Z_{\Fext}} \omega = \Fext\, , \quad \contr{Z_{\Fext}} \eta = 0 \, .$$
or in coordinates
$$Z_{\Fext} = - \Fext_i \frac{\partial}{\partial p_i}\, .$$
Consequently,
\begin{equation}\label{eq:forced_evolution_contractions}
    \contr{\evol_{H, \Fext}} \omega = \dd H - \Reeb(H) \eta - \Fext \, , \quad \contr{\evol_{H, \Fext}} \eta = \Reeb(H) +1 \, .
\end{equation}




\subsubsection*{The Lagrangian formalism}

In the Lagrangian formalism, the configuration of an adiabatically closed simple thermodynamic system will be determined by a point in $\T Q \times \RR$. The dynamics will be derived from a Lagrangian function $L\in \Cinfty(\T Q \times \RR)$.
The \emph{Legendre transform} $\Leg\colon \T Q\times \RR \to \cT Q \times \RR$ is locally given by
$$\Leg(q^i, v^i, S) = \left(q^i, \pder{L}{v^i}, S\right)\, ,$$
in canonical bundle coordinates $(q^i, v^i, S)$ of $\T Q \times \RR$ and $(q^i, p_i, S)$ of $\cT Q \times \RR$.

The \emph{Lagrangian energy} is the function $E_L = \Delta(L) - L\in \Cinfty(\T Q\times \RR)$, where $\displaystyle{\Delta = v^i \partial_{v^i}}$ is the Liouville vector field (i.e., the infinitesimal generator of homotheties on the fibers of $\T Q$). Let us recall that $L$ is regular, i.e., $\Leg$ is a local diffeomorphism if and only if  the Hessian matrix $$ \left(W_{ij}\right) = \left({\frac{\partial^2L}{\partial v^i \partial v^j}}\right)$$
is non-singular. We shall assume that $L$ is hyper-regular, namely that $\Leg$ is a (global) vector bundle isomorphism. (If $L$ is not hyper-regular, it suffices to work in an open subset of $\T Q\times \R$ such that the restriction of $\Leg$ to that subset is a diffeomorphism.) Then, from $E_L$ we can define the Hamiltonian function $H = E_L \circ \Leg^{-1}$. In bundle coordinates,
$$H(q,p,S) = p_i v^i (q,p,S)-L(q,v(q,p,S),S)\, ,$$
where $v^i(q, p, S)$ denotes the coordinate $v^i$ of $\Leg^{-1} (q, p, S)$. 
Therefore,
\begin{align*}
    &\pder{H}{q^i} = p_j \pder{v^j}{q^i} - \pder{L}{q^i} - \pder{L}{v^j}\pder{v^j}{q^i} = - \pder{L}{q^i} \, , \\
    &\pder{H}{p_i}= v^i+p_j\pder{v^j}{p_i}-\pder{L}{v^j}\pder{v_j}{p_i}=v^i
    &\pder{H}{S} = p_j \pder{v^j}{S} - \pder{L}{S} - \pder{L}{v^j}\pder{v^j}{S} = - \pder{L}{S}\, .
\end{align*}
The second of the previous equations and the first equation of~\eqref{CurveAdiabaticallyClosedSimple} imply that
\begin{equation}\label{sode}
    \frac{\dd q^i}{\dd t}= v^i.
\end{equation}
The second of the equations~\eqref{CurveAdiabaticallyClosedSimple} then implies that
$$
\frac{\dd}{\dd t} \left(\pder{L}{v^i}\right) =  \pder{L}{q^i}+\Ffr_i+\Fext_i\, ,
$$
and the combination of the first and third of the equations~\eqref{CurveAdiabaticallyClosedSimple} yields
$$\pder{L}{S}\frac{\dd S}{\dd t} = \frac{\dd q^i}{\dd t} \Ffr_i\, .$$
Note that these last two differential equations are for functions on the parameters $(q^i, p_i, S)$. In order to obtain the corresponding differential equations for the parameters $(q^i, v^i, S)$, it suffices to compose with the Legendre transformation (which we have assumed to be a diffeomorphism), obtaining
\begin{equation}\label{eq:thermodynamic_Euler_Lagrange}
\begin{aligned} 
        &\frac{\dd}{\dd t}\left(\pder{L}{v^i}\right)-\pder{L}{q^i}=\Ftfr_i+\Ftext_i, \\
        &\pder{L}{S}\frac{\dd S}{\dd t}=\frac{\dd q^i}{\dd t}\Ftfr_i,
\end{aligned}
\end{equation}
where, by means of the Legendre transform, we define the following differential forms on $\T Q\times \RR$:
\begin{equation}\label{eq:forms_Lagrangian}
\begin{array}{lll}
\theta_L = \Leg^\ast \theta\, , \quad &\omega_L = -\dd \theta_L = \Leg^\ast \omega\, , \quad &\eta_L = \Leg^\ast \eta\, , \\\\
\Ftext = \Leg^\ast \Fext\, , \quad &\Ftfr = \Leg^\ast \Ffr\, ,
\end{array}
\end{equation}
whose local expressions are
\begin{equation}\label{eq:forms_Lagrangian_local_expressions}
\theta_L = \frac{\partial L}{\partial v^i} \dd q^i\, , \quad
\eta_L=\frac{\partial L}{\partial S} \dd S-\Ftfr\, , \quad
\Ftfr = \Ftfr_i \, \dd q^i \, , \quad
\Ftext = \Ftext_i \, \dd q^i \, .
\end{equation}
We shall refer to equations \eqref{eq:thermodynamic_Euler_Lagrange} as the \emph{thermodynamic Euler--Lagrange equations}. The tuple $(\T Q\times \RR, \omega_L, \eta_L, L, \Ftext, \Ftfr)$ will be called the \emph{Lagrangian thermodynamic system}.

One can easily see that $(\omega_L, \eta_L)$ is a partially cosymplectic structure on $\T Q \times \RR$. Indeed, since $\Leg$ is a vector bundle isomorphism, it maps $\ker \eta$ into $\ker \eta_L$ and $\ker \omega$ into $\ker \omega_L$ (see Proposition~\ref{proposition:equivalent_defs_almost_cosymp}).
Let $\sode_{(L,\, \Ftext)}=\evol_{(E_L,\, \Ftext)}$ be the evolution vector field of $E_L$ subject to the external forces $\Ftext$, defined by the equation
$$
    \flat_{(\Omega_L,\eta_L)}(\sode_{(L,\, \Ftext)})=\dd E_L+\eta_L-\Ftext\, .
$$
It is related with the evolution vector field of $H$ subject to the external forces $\Fext$ as follows:
$$\evol_{(H,\, \Fext)} =\T \Leg\, \circ\, \sode_{(L,\, \Ftext)} \circ \Leg^{-1}\, .$$
Hence, $\gamma\colon I \subseteq \RR \to \T Q \times \RR$ is an integral curve of $\sode_{(L,\, \Ftext)}$ if and only if $\Leg^{-1} \circ\, \gamma$ is an integral curve of $\evol_{(H,\, \Fext)}$. Furthermore, $\gamma(t) = \big(q^i(t), v^i(t), S(t)\big)$ is an integral curve of $\sode_{(L,\, \Ftext)}$ if and only if it satisfies the system of differential equations~\eqref{eq:thermodynamic_Euler_Lagrange}.

\bigskip 
\noindent It is noteworthy that equations~\eqref{eq:thermodynamic_Euler_Lagrange} are equivalent to those obtained by Gay-Balmaz and Yoshimura \cite{G.Y2018,Y.G2022}.


\begin{remark}\label{remark:temperature}
    Following \cite{G.Y2018}, in this formulation, the temperature is defined as minus the derivative of $L$ with respect to $S$, i.e., $T=-\frac{\partial L}{\partial S}$, which is assumed to be positive. When the Lagrangian has the standard form
    $$
    L(q, v, S)=K(q, v)-U(q, S)\, ,
    $$
    where the kinetic energy $K$ is assumed to be independent of $S$, and $U(q, S)$ is the internal energy, then $T=-\frac{\partial L}{\partial S}=\frac{\partial U}{\partial S}$ recovers the standard definition of the temperature in thermodynamics.
\end{remark}

\begin{remark}\label{remark:RelationToContact}
    Previous studies of thermodynamics have been carried out using contact geometry. In \cite{A.d.L+2020}, given a function $H\in \Cinfty(M)$, a contact evolution field $Y_H\in \X(M)$ is defined as
    $$
    \flat_{(\eta_C,\, \dd \eta_C)}(Y_H)= \dd H-\Reeb_C(H)\eta_C\, ,
    $$
    where $\eta_C$ is a contact form on $M$, with Reeb vector field $\Reeb_C$. In the case where $M=\cT Q\times \RR$, with canonical bundle coordinates $(q^i, p_i, S)$, these objects read
    $$\eta_C=\dd S-p_i\dd q^i\, , \quad \Reeb_C = \pder{}{S}\, , \quad Y_H = \pder{H}{p_i} \pder{}{q^i} - \left(\pder{H}{q^i} + p_i \pder{H}{S}\right) \pder{}{p_i} + p_i \pder{H}{p_i} \pder{}{S}\, .$$
    On the other hand, we can endow $M$ with the partially cosymplectic structure $(\omega,\eta)$ given by \eqref{eq:local_forms_Hamiltonian}, for $\Ffr_i=-\Reeb_C(H)p_i$. In this way, we can write $Y_H = \evol_{H, \Fext}$, with $\Fext\equiv 0$. This can be easily verified by comparing the local expression of $Y_H$ above with \eqref{eq:evol_H_F_coords}, the local expression of $\evol_{H, \Fext}$.
    Thus, the evolution vector defined in \cite{A.d.L+2020} can be regarded as a particular case of the one defined in this paper. Notice that, since $\eta\neq \eta_C$, the isomorphisms $ \flat_{(\eta_C,\, \dd \eta_C)}$ and $\flat_{(\eta,\, \dd \eta)}$ are different.

     \noindent Similarly, the thermodynamical Herglotz equations for a contact Lagrangian from \cite{A.d.L+2020} can be recovered from our thermodynamic Euler--Lagrange equations \eqref{eq:thermodynamic_Euler_Lagrange}. Indeed,
     taking $\Ftfr_i = \pder{L}{S}\pder{L}{v_i}$ and $\Ftext\equiv 0$, we have that
     \begin{equation}
    \begin{aligned} 
            &\frac{\dd}{\dd t}\left(\pder{L}{v^i}\right)-\pder{L}{q^i}-\pder{L}{S} \pder{L}{v_i}=0\, , \\
            &\frac{\dd S}{\dd t}=\frac{\dd q^i}{\dd t} \pder{L}{v_i}\, .
    \end{aligned}
    \end{equation}
\end{remark}

\subsubsection*{SODEs}

As we have proven above, the solutions of the equations (\ref{eq:thermodynamic_Euler_Lagrange}) satisfy the equation (\ref{sode}), that is, the corresponding dynamics is a Second Order Differential Equation (SODE, for short). We shall now recall some notions about SODEs.

Let us recall that the tangent bundle $\T Q$ of a manifold $Q$ is canonically endowed with the \emph{vertical endomorphism} $\Sendo$, a $(1, 1)$ tensor field which locally reads
$$\Sendo =  \frac{\partial}{\partial v^i} \otimes \dd q^i\, ,$$
in canonical bundle coordinates $(q^i, v^i)$, i.e.,
$$\Sendo \left(\frac{\partial}{\partial q^i}\right) = \frac{\partial}{\partial v^i}\, , \quad \Sendo\left(\frac{\partial}{\partial v^i}\right) = 0\, .$$
In standard Lagrangian mechanics, given a Lagrangian $L: \T Q \to \mathbb{R}$ 
is regular if and only if the $2$-form $\omega_L = - \dd \big(\Sendoadj(\dd L)\big)$ is symplectic. In that case, we obtain the dynamics are given by the unique vector field $\xi_L$ satisfying the equation
$$
\contr{\xi_L} \omega_L = \dd E_L\, ,
$$
where $E_L = \Delta (L) - L$, with $\Delta$ the Liouville vector field.

The vector field $\xi_L$ above is an example of \emph{SODE}, sometimes called a semispray or semigerbe.
A SODE $\xi$ is a vector field on $\T Q$ such that its integral curves are just the tangent lifts of its projections to $\T Q$. Thus, in bundle coordinates $(q^i, v^i)$ on $\T Q$, a SODE reads
$$
\xi = v^i \, \frac{\partial}{\partial q^i} + \xi^i(q,v) \, \frac{\partial}{\partial v^i}\, ,
$$
for some function $\xi^i$.
One immediately can see that the above condition is equivalent to
$$
\Sendo (\xi) = \Delta\, .
$$
Then, its integral curves $(q^i (t), v^i(t))$ satisfy the following system of differential equations
\begin{eqnarray*}
    && \frac{\dd q^i}{\dd t} = v^i\, ,\\
    && \frac{\dd v^i}{\dd t} = \xi^i(q,v)\, ,
\end{eqnarray*}
and thereupon
\begin{eqnarray*}
    && \frac{\dd q^i}{\dd t} = v^i\, ,\\
    && \frac{\dd^2q^i}{\dd t^2} = \xi^i\left(q,\frac{\dd q}{\dd t}\right)\, ,
\end{eqnarray*}
which justifies the name of SODE.

We can extend this notion to the time-dependent case as a vector field $\xi$ on $\mathbb{R} \times \T Q$ satisfying the condition
$$
\Sendobar (\xi) = \frac{\partial}{\partial t} + \Delta,
$$
where 
$$
\Sendobar = \dd t \otimes \frac{\partial}{\partial t}  + \Sendo
$$ 
is the so-called stable tangent structure on $\mathbb{R} \times \T Q$ (see \cite{poincare}). Such a vector field is called a Non-autonomous Second Order Differential Equations (NSODE, for short) and its local expression is as follows:
$$
\xi = \frac{\partial}{\partial t} + v^i \, \frac{\partial}{\partial q^i} + \xi^i(t,q,v) \, \frac{\partial}{\partial v^i}.
$$
This is just the type of vector field that appears in the solution of the Euler--Lagrangian equations for non-autonomous (or time-dependent) Lagrangians on $\mathbb{R} \times \T Q$.

However, a different situation occurs when we are dealing with contact Lagrangian systems on $\T Q \times \mathbb{R}$. In that case, we are interested in vector fields $\xi$ on $\T Q \times \mathbb{R}$ satisfying the equation
$$
\Sendo(\xi) = \Delta,
$$
where, slightly abusing of notation, we denote by $\Sendo$ the trivial lift to $\T Q \times \RR$ of the canonical vertical endomorphism on $\T Q$. Then, the local expression in induced coordinates $(q^i, v^i, z)$ is
$$
\xi = \mathcal{A}(q,v,z) \, \frac{\partial}{\partial z} + v^i \, \frac{\partial}{\partial q^i} + \xi^i(q,v,z) \, \frac{\partial}{\partial v^i},
$$
for some functions $\mathcal{A}$ and $\xi^i$. In the case of a contact Lagrangian system we get
$$
\xi = L(q,v,z) \, \frac{\partial}{\partial z} + v^i \, \frac{\partial}{\partial q^i} + \xi^i(q,v,z) \, \frac{\partial}{\partial v^i},
$$
where $L\colon \T Q \times \RR\to \RR$ is the Lagrangian function (see \cite{A.d.L+2020}). However, in the current paper that function $\mathcal{A}$ is more general since the system can also include external and friction forces.

\noindent Regarding $\Sendo$ as an endomorphism on the vector bundle $\T(\T Q \times \RR)$, we define $\Sendoadj$ as the adjoint endomorphism on the dual bundle $\cT (\T Q \times \RR)$.
Using the local expression~\eqref{eq:forms_Lagrangian_local_expressions}, it is straightforward to verify that $\theta_L$ can be written as
$$\theta_L = \Sendoadj(\dd L)\, ,$$
which implies that
$$\omega_L = - \dd \big(\Sendoadj(\dd L)\big)\, ,$$
and
\begin{equation}\label{eq:contraction_theta_endomorphism}
    \contr{X}\theta_L = \liedv{\Sendo\circ X} L\, , \quad \forall X\in \X(\T Q \times \RR)\, .
\end{equation}

\subsubsection*{Evolution of the energy and the geometric structures}
The local expression~\eqref{eq:evol_H_F_coords} of the forced evolution vector field implies that
\begin{equation}
    \evol_{H, \Fext} (H)
    = \Fext_i \pder{H}{p_i}\, .
\end{equation}
Therefore, by integrating along the flow $\phi_t$ of $ \evol_{H, \Fext}$, we obtain
\begin{equation}
    H \circ \phi_t \left(q^i, p_i, S\right) = H \left(q^i, p_i, S\right) + \int_0^t  \left(\Fext_i \circ \phi_\tau \left(q^i, p_i, S\right)\right) \left(\pder{H}{p_i} \circ \phi_\tau \left(q^i, p_i, S\right)\right)\, \dd \tau\, .
\end{equation}
In particular, the Hamiltonian is a constant of the motion if
$\Fext\equiv 0$. Since the inverse Legendre transform maps $H$ to $E_L$ and $\phi_t$ to the flow $\psi_t$ of $\sode_{(L,\, \Ftext)}$, we have an analogous expression in the Lagrangian counterpart:
\begin{equation}
    E_L \circ \psi_t \left(q^i_0, v^i_0, S_0\right) = E_L \left(q^i_0, v^i_0, S_0\right) + \int_0^t  \left( \Ftext_i \circ \psi_\tau \left(q^i_0, v^i_0, S_0\right)\right) \left(v^i \circ \psi_\tau \left(q^i_0, v^i_0, S_0\right)\right)\, \dd \tau\, .
\end{equation}

The analysis of the relations between the evolution field and the geometric structure allows us to give an equivalent characterization of it.

Let $H$ be a Hamiltonian function on $M$, subject to the external forces $\Fext$ and to the friction forces. Consider the partially cosymplectic structure $(\omega,\eta)$ on $M$ derived from it, with local expressions \eqref{eq:local_forms_Hamiltonian}. It is easy to see (for instance, using the local expressions) that the evolution vector field subject to external forces $\evol_{H, \Fext}$  satisfies that
$$\eta(\evol_{H, \Fext})=0\, .$$
From Proposition~\ref{proposition:equivalent_defs_almost_cosymp}, it is clear that the isomorphism $\flat_{(\omega,\, \eta)}\colon \X(M)\to \Omega^1(M)$ can be decomposed into the morphisms $X\mapsto \eta(X) \eta$ and $X\mapsto \contr{X} \omega$, which have ranks $1$ and $2n$, respectively. This implies the following.

\begin{proposition}
   The evolution vector field $\evol_{H, \Fext}$ of $H$ subject to external forces $\Fext$ is the unique vector field over $M$ such that
   $$\evol_{H,\Fext}\in \ker \eta\, , \quad \contr{\evol_{H,\Fext}}\omega=\dd H+\eta-\Fext\, .$$
   Consequently,
   $$\liedv{\evol_{H,\Fext}} \omega =\dd\eta-\dd\Fext\, , \quad \liedv{\evol_{H,\Fext}}\eta=\contr{\evol_H}\dd\eta\, .$$
\end{proposition}

\subsubsection*{Infinitesimal Symmetries}

Let $X$ be a vector field on $Q$, and denote by $X^C$ and $X^V$ its complete and vertical lifts, respectively, to the tangent bundle $\T Q$. We will denote by the same symbols their trivial extensions to $\T Q \times \RR$. If $(q^i)$ are local coordinates in $Q$, which induce bundle coordinates $(q^i, v^i, S)$ on $\T Q \times \RR$, and
$$
X = X^i \, \frac{\partial}{\partial q^i}
$$
then
$$
X^C = X^i \, \frac{\partial}{\partial q^i} + v^j \frac{\partial X^i}{\partial q^j} \, \frac{\partial}{\partial v^i}\, , \quad
X^V = X^i \, \frac{\partial}{\partial v^i}\, .
$$
See \cite{yano-ishihara,de2011methods} for the intrinsic definitions and additional details.

\begin{theorem}[Noether' s theorem]\label{NoetherI}
    Let $(\T Q\times \RR, \omega_L, \eta_L, L, \Ftext, \Ftfr)$ be a Lagrangian thermodynamic system.
    Let $X$ be a vector field on $Q$. Then,
    \begin{equation}\label{condition}
    X^C(L) = - (\Ftfr+\Ftext)(X^C)
    \end{equation} 
   if and only if $X^V(L)$ is a conserved quantity.
\end{theorem}
\begin{proof}
    We will prove the result using local coordinates. Along an integral curve $\gamma(t)=(q^i(t), v^i(t), S(t))$ of $\sode_{(L,\, \Ftext)}$, we have that
    \begin{align*}
        \frac{\dd}{\dd t}\left(X^V(L)\right)&=
        \frac{\partial X^i}{\partial q^j} v^j\frac{\partial L}{\partial v^i} + X^i \frac{\dd}{\dd t} \frac{\partial L}{\partial v^i} =  \frac{\partial X^i}{\partial q^j} v^j\frac{\partial L}{\partial v^i} + X^i \left(\pder{L}{q^i}+\Ftfr_i+\Ftext_i\right)\\
        &= X^C(L) + (\Ftfr+\Ftext)(X^C)\, ,
    \end{align*}
    where we have used the first equation in \eqref{eq:thermodynamic_Euler_Lagrange} in the second step.
\end{proof}

\noindent We may also consider infinitesimal symmetries on $\cT Q \times \RR$ which are not lifted from $Q$.

\begin{proposition}
    Let $(\cT Q \times \RR, \omega, \eta, H, \Fext, \Ffr)$ be a Hamiltonian thermodynamic system.
    Let $X$ be a vector field on $\cT Q \times \RR$ such that $\liedv{X} \theta$ is exact (where $\theta$ is the pullback of the canonical one-form on $\cT Q$).
    Then, for any function $f\in \Cinfty(\cT Q \times \RR)$ such that $\dd f = \liedv{X} \theta$, the following statements are equivalent:
    \begin{enumerate}
        \item $g\coloneqq f - \theta(X)$ is a conserved quantity,
        \item $X(H) + \eta (X) - \Fext(X)=0$.
    \end{enumerate}
\end{proposition}

\begin{proof}
    By Cartan's magic formula, $\liedv{X} \theta = \dd f$ if and only if
    $$\contr{X} \omega = - \contr{X} \dd \theta = \dd \left(\contr{X} \theta - f\right) = -\dd g\, .$$
    Contracting both sides with the evolution vector field $\evol_{H, \Fext}$ of $H$ subject to external forces $\Fext$ yields
    $$\liedv{\evol_{H, \Fext}} g = \contr{X} \contr{\evol_{H, \Fext}}  \omega_Q = X(H) - \Reeb(H) \eta(X) - \Fext(X)\, ,$$
    where we have used equation~\eqref{eq:forced_evolution_contractions} in the last step. In particular, the left-hand side vanishes if and only if $g$ is preserved along the flow of $\evol_{H, \Fext}$.
\end{proof}

\noindent The Lagrangian counterpart of the proposition above is as follows:
\begin{proposition}\label{proposition:Cartan_symmetries_Lagrangian}
    Let $(\T Q \times \RR, \omega_L, \eta_L, L, \Ftext, \Ftfr)$ be a Lagrangian thermodynamic system.
    Let $X$ be a vector field on $\T Q \times \RR$ such that $\liedv{X} \theta_L$ is exact, where $\theta_L=\Leg^\ast \theta$. Then, for any function $f\in \Cinfty(\T Q \times \RR)$ such that $\dd f = \liedv{X} \theta_L$, the following statements are equivalent:
    \begin{enumerate}
        \item $g\coloneqq f - \theta_L(X)$ is a conserved quantity,
        \item $X(E_L) + \eta_L (X) - \Ftext(X)=0$.
    \end{enumerate}
\end{proposition}

\noindent Using the identity \eqref{eq:contraction_theta_endomorphism}, we can also express the conserved quantity $g$ above as 
$$g= f - (\Sendo \circ X) (L)\, ,$$
with $\Sendo$ is the trivial lift to $\T Q \times \RR$ of the canonical vertical endomorphism on $\T Q$.

\begin{example}[A cylinder with two pistons]
    Consider a cylinder with a piston on each end, filled with an ideal gas. The configuration space of the system is $Q=\RR^2$, with canonical coordinates $(x, y)$. Let $(x, y, v_x, v_y, S)$ be the induced coordinates on $\T Q \times \RR$.
    In suitable units and for a suitable number of particles in the gas, the Lagrangian that describes this system is given by
    \begin{equation*}
        L(x, y, v_x, v_y, S)=\frac{v_x^2}{2} + \frac{v_y^2}{2}-e^S(x+y)^{-\frac{1}{c}},
    \end{equation*}
    where $c=\frac{3}{2}$ is the molar specific heat capacity at constant volume divided by $R$, the constant of ideal gases, for a monoatomic gas. We also consider a friction force between the cylinder and the pistons, which is proportional to its velocity. In this case, we consider a constant temperature independent coefficient of friction $\gamma$:
    \begin{equation*}
        \Ftfr=-\gamma v_x \dd x - \gamma v_y \dd y\, ,
    \end{equation*}
    and no external forces, namely, $\Ftext \equiv 0$.
    The Lagrangian energy $E_L = \Delta(L)-L$ reads
    \begin{equation*}
        E_L(x, y, v_x, v_y, S)=\frac{v_x^2}{2} + \frac{v_y^2}{2}+e^S(x+y)^{-\frac{1}{c}}\, ,
    \end{equation*}
    the one-form $\theta_L$ is given by
    $$\theta_L=\frac{\partial L}{\partial v_x} \dd x + \frac{\partial L}{\partial v_y} \dd y = v_x \dd x + v_y \dd y\, ,$$
    and the one-form $\eta_L$ is given by
    $$\eta_L=\frac{\partial L}{\partial S} \dd S-\Ftfr = \gamma v_x \dd x + \gamma v_y \dd y -e^S(x+y)^{-\frac{1}{c}}\dd S \, .$$
    Consider the vector field 
    $$X = \partial_y - \partial_x + \gamma \partial_{v_x} - \gamma \partial_{v_y}\, .$$
    Then, we have
    $$\Sendo \circ X = \partial_{v_y} - \partial_{v_x}\, .$$
    Thus,
    $$\theta_L(X) = (\Sendo \circ X) (L) = v_y - v_x\, .$$
    Observe that
    $$X(E_L) = \gamma \left(v_x - v_y\right) = - \eta_L(X)\, ,$$
    and
    $$\liedv{X} \theta_L = \gamma \dd x - \gamma \dd y = \dd f\, , \quad f \coloneqq \gamma(x-y) \, .$$
    Hence, by Proposition~\ref{proposition:Cartan_symmetries_Lagrangian}, 
    the function
    $$g\coloneqq f - \theta_L(X) = \gamma(x-y) + v_y -v_x$$
    is a conserved quantity.
    
\end{example}

\section[Discrete model]{Discrete model for adiabatically closed simple \\ thermodynamic systems}\label{sec:discrete_model}

\subsection{Discrete variational principle}\label{sec:discrete_variational_principle}

In this section, we will construct a discrete model inspired by the continuous one presented in the previous section. For simplicity's sake, we will first consider the case without external forces. Discrete external forces will be added at the end of the present subsection.

Consider a simple thermodynamic system whose configuration manifold is $Q$ and let $S\in \RR$ be its entropy. A $\Cinfty$ function $L_d:Q\times Q\times \RR\longrightarrow \RR$ will be called a \emph{discrete Lagrangian}. Consider a $1$-form $\ffrd=\parent{\ffrdm,\ffrdp,0}\in \Omega^1(Q\times Q\times \RR)$ that will represent the discrete version of the friction forces considered above. Here we are using the natural identification $\cT(Q\times Q\times \RR) \cong \cT Q \oplus \cT Q \oplus \cT \RR$.
The triple $\parent{Q,L_d,\ffrd}$ will be called a \emph{simple discrete thermodynamic system}. We will assume that $\DD_S L_d\neq 0$ at every point, where $\DD_S L_d$ denotes the derivative with respect to the real variable $S$ of the function $L_d$.
We will regard $\DD_S L_d$ as a discrete approximation to the temperature of the system (\textit{cf}.~Remark~\ref{remark:temperature}).
Since the third principle of thermodynamics assures that no system can reach zero absolute temperature, this assumption is reasonable.

Let $N\in\NN$ and $\{t_k=hk\mid k=0,\ldots,N\}$ be an increasing sequence that represents the sequence of times. Hereinafter, $h$ will be called the \emph{time step}. 

\begin{definition}
    Under the previous conditions we define the \emph{total discrete path space} as the set
    $$
    \Omega\parent{\{t_k\}_{k=0}^N}=\left\{\sigma:\{t_k\}_{k=0}^N\longrightarrow Q\times \RR\right\}.
    $$
    We will denote $\sigma(t_k)=(q_k,S_k)$, $k=0,\ldots, N$. The \emph{discrete thermodynamic path space} is the subset of $\Omega\parent{\{t_k\}_{k=0}^N}$ given by
    \begin{align*}
     &  \Omega_T\parent{\{t_k\}_{k=0}^N} =
    \Bigg\{\sigma\in  \Omega\parent{\{t_k\}_{k=0}^N}\, \mid  \\
    &\qquad S_k-S_{k-1}=\frac{\ffrdp \parent{q_{k-1},q_k,S_{k-1}}q_{k}-\ffrdm \parent{q_{k-1},q_k,S_{k-1}}q_{k-1}}{\DD_S L_d\parent{q_{k-1},q_k,S_{k-1}}}, 1\leq k\leq N \Bigg\},
    \end{align*}
    where
    $$\ffrdpm  q_k=\sum_{i=1}^n f_{d,i}^{\mathrm{fr},\pm}q_k^i\, ,$$
    with $(q_k^i)$ the coordinates of the point $q_k\in Q$ on a certain chart $(q^i)$, in which the forces read
    $$\ffrdpm   = \sum_{i=1}^n f_{d,i}^{\mathrm{fr},\pm} \dd q^i\, .$$
    Given two fixed points $q^*_0,q^*_N\in Q$, the discrete thermodynamic path space from $q_0^*$ to $q_N^*$ is the set
    $$
    \Omega_T\parent{q^*_0,q^*_N,\{t_k\}_{k=0}^N}=\left\{\sigma\in  \Omega_T\parent{\{t_k\}_{k=0}^N}\mid q_0=q_0^*, q_N=q_N^* \right\}.
    $$
\end{definition}

\begin{remark}
    Notice that the equations defining the discrete thermodynamic path space are a discretization of the phenomenological constraints used in \cite{G.Y2018}. We should also observe that, in general, they depend on the coordinate chart chosen, since $\ffrdpm  q_k$ does. However, if $Q$ is a vector space, we can choose linear coordinates and take $\ffrdp =\ffrdm$. In that case, the definition above becomes coordinate-independent. 
\end{remark}

The total discrete path space $\Omega\parent{\{t_k\}_{k=0}^N}$ and the discrete thermodynamic path space $\Omega_T\parent{\{t_k\}_{k=0}^N}$ are diffeomorphic to the product manifolds $Q^N\times \RR^N$ and $Q^N\times\RR$, respectively. Thus, their tangent spaces can be identified with the Whitney sums $(\T Q)^N\oplus (\T\RR)^N$ and $(\T Q)^N\oplus \T \RR$, respectively. An element $\delta \sigma$ of $\T \Omega\parent{\{t_k\}_{k=0}^N}$ is called a \emph{variation}.

\begin{definition}
    We define the \emph{discrete action} of the thermodynamic system as the map
    \begin{align*}
        \mathcal{A}_d:\Omega\parent{\{t_k\}_{k=0}^N}&\longrightarrow \RR\\
        \sigma&\longmapsto \mathcal{A}_d(\textbf{q})=\sum_{k=1}^NL_d\parent{q_{k-1},q_k,S_{k-1}} \, ,
    \end{align*}
    where $\sigma (\{t_k\})=\parent{q_k,S_k}$ for every $k=0,\ldots, N$.
\end{definition}

\begin{theorem}\label{theorem:discrete_EL_map}
    Given a simple discrete thermodynamic system $\parent{Q,L_d,\ffrd}$, we define the $1$-form
    $$
    \eta_d=\parent{\frac{1}{2}\ffrdm ,\frac{1}{2}\ffrdp ,-\DD_S L_d}\in \Omega^1(Q\times Q\times \RR)\, ,
    $$
    and consider the distribution $\mathcal{D}$ on $\Omega\parent{\{t_k\}_{k=0}^N}$ given by
    $$
    \mathcal{D}=\{\delta \sigma =   \parent{\delta q_0,\ldots, \delta q_N,\delta S_0,\ldots, \delta S_N}_{(q_0,\ldots,q_N,S_0,\ldots,S_N)}\in \T\Omega\parent{\{t_k\}_{k=0}^N}\; \mid
    $$
    $$\eta_d\parent{q_{k-1},q_k,S_{k-1}}\parent{\delta q_{k-1},\delta q_k,\delta S_{k-1}}=0, k=1,\ldots, N\}.
    $$
    There exists a unique map $\DD_{\mathrm{DEL}}L_d:Q^3\times\RR^2\longrightarrow \cT Q$ and unique 1-forms on $Q\times Q\times \RR$, $\Theta^+_{L_d},\; \Theta^-_{L_d}$ such that, for all variations on $\delta \sigma \in \mathcal{D}$, we have
    \begin{equation}\label{DiferencialAccionTotal}
    \begin{aligned}
         \dd\mathcal{A}_d(\delta \sigma)& =\sum_{k=1}^{N} \DD_{\mathrm{DEL}}L_d\parent{q_{k-1},q_k,q_{k+1},S_{k-1},S_k} \delta q_k 
         \\ & \quad -\Theta^-_{L_d}\parent{q_0,q_1,S_0}\parent{\delta q_0,\delta q_1,\delta S_0}
         \\ & \quad +\Theta^+_{L_d}\parent{q_{N-1},q_N,S_{N-1}}\parent{\delta q_{N-1},\delta q_N,\delta S_{N-1}},
    \end{aligned}
    \end{equation}
    where we identify $\delta \sigma= \parent{\delta q_0,\ldots, \delta q_N,\delta S_0,\ldots, \delta S_N}\in (\T Q)^N\oplus (\T \RR)^N$.
    These objects have the following coordinate expressions:
    \begin{align*}
        & \DD_{\mathrm{DEL}}L_d\parent{q_{k-1},q_k,q_{k+1},S_{k-1},S_k}
        =\Big(\DD_1L_d\parent{q_k,q_{k+1},S_k}+\frac{1}{2}\ffrdm\parent{q_k,q_{k+1},S_k}\\
        & \qquad \qquad
        +\DD_2L_d\parent{q_{k-1},q_{k},S_{k-1}}+\frac{1}{2}\ffrdp\parent{q_{k-1},q_{k},S_{k-1}}\Big)\cdot \dd q\, ,\\
        & \Theta^-_{L_d}\parent{q_0,q_1,S_0}=-\left(\DD_1L_d\parent{q_0,q_{1},S_0}+\frac{1}{2}\ffrdm\parent{q_0,q_{1},S_0}\right)\cdot \dd q_0\, ,\\
        & \Theta^+_{L_d}\parent{q_{N-1},q_N,S_{N-1}}=\left(\DD_2L_d\parent{q_{N-1},q_{N},S_{N-1}}+\frac{1}{2}\ffrdp\parent{q_{N-1},q_{N},S_{N-1}}\right)\cdot \dd q_N\, .
    \end{align*}
\end{theorem}
\begin{proof}
    Identifying $\Omega\parent{\{t_k\}_{k=0}^N}$ with $(\T Q)^N\oplus (\T\RR)^N$, the differential of the discrete action is given by
    \begin{align*}
        \dd\mathcal{A}_d&=\sum_{k=1}^N \Big(\DD_1L_d\parent{q_{k-1},q_k,S_{k-1}}\dd q_{k-1}+\DD_2L_d\parent{q_{k-1},q_k,S_{k-1}}\dd q_{k}\\
        & \quad +\DD_S L_d\parent{q_{k-1},q_k,S_{k-1}}\dd S_{k-1}\Big),
    \end{align*}
    where $\DD_1$ and $\DD_2$ denote the exterior differential with respect to the first and second variables of the product manifold, respectively. When restricted to acting over variations on $\mathcal{D}$, for every $k=1, \ldots, N$ we obtain that
    \begin{equation}\label{1formaDiscreta}  
    \DD_S L_d\parent{q_{k-1},q_k,S_{k-1}}\dd S_{k-1}=\frac{1}{2}\parent{\ffrdm\parent{q_{k-1},q_k,S_{k-1}}\dd q_{k-1}+\ffrdp\parent{q_{k-1},q_k,S_{k-1}}\dd q_{k}}.   
    \end{equation}
    Thus,
    $$
    \restr{\dd\mathcal{A}_d}{\mathcal{D}}=\sum_{k=1}^N \bigg(\DD_1L_d\parent{q_{k-1},q_k,S_{k-1}}\dd q_{k-1}+\DD_2L_d\parent{q_{k-1},q_k,S_{k-1}}\dd q_{k}+
    $$
    $$
    \left.+\frac{1}{2}\ffrdm\parent{q_{k-1},q_k,S_{k-1}}\dd q_{k-1}+\frac{1}{2}\ffrdp\parent{q_{k-1},q_k,S_{k-1}}\dd q_{k}\right).
    $$
    Grouping the terms accompanying each $\dd q_{k}$, we get the result.
\end{proof}

It is noteworthy that, under the assumption that $\DD_S L_d\neq 0$ at the points of the discrete path, $\{\dd q_{0},\ldots, \dd q_{N}\}$ is a basis at every point of the dual space of $\mathcal{D}|_{(q_0,\ldots,q_N,S_0,\ldots,S_N)}$. 
It is also worth recalling that each $q_k$ is a point in $Q$ and not a coordinate. In general, $Q$ will be an $n$-dimensional manifold so that, for each of its copies, we will have coordinates that will be denoted by $q_k^i$, $i=1,\ldots, n$, $k=0,\ldots, N$. The same can be said for the $1$-forms $\dd q_{k}$ and the variations $\delta q_k$. Nonetheless, each $S_k$ is a real variable.

\begin{remark}
    Equation \eqref{1formaDiscreta} can be interpreted as the discretization of the variational constraint introduced in \cite{G.Y2018}, via the identification
    $$
    \dd q\approx  \frac{\dd q_{k-1}+\dd q_k}{2}.
    $$
\end{remark}

\begin{definition}
    A discrete thermodynamic path $\sigma\in \Omega_T\parent{\{t_k\}_{k=0}^N}$ is said to be a \emph{solution of the discrete thermodynamic Euler--Lagrange equations} if for all variations $\delta\sigma\in \mathcal{D}_\sigma$ with fixed extreme points, i.e., $\delta q_0=0=\delta q_N$, we have
    $$
    \dd\mathcal{A}_d\cdot \delta \sigma  =0.
    $$
\end{definition}

Using the previous theorem, we can conclude the following characterization of the solutions of the discrete thermodynamic Euler--Lagrange equations.

\begin{corollary}
    A discrete path $\sigma\in \Omega\parent{\{t_k\}_{k=0}^N}$ is a solution of the discrete thermodynamic Euler--Lagrange equations if an only if it satisfies
    \begin{subequations}\label{eq:EulerLagrangeDiscretas}
    \begin{equation} \label{EulerLagrangeTermoDiscreta}
     \DD_1L_d\parent{q_k,q_{k+1},S_k}+\frac{1}{2}\ffrdm\parent{q_k,q_{k+1},S_k}+\DD_2L_d\parent{q_{k-1},q_{k},S_{k-1}}+\frac{1}{2}\ffrdp\parent{q_{k-1},q_{k},S_{k-1}}=0,   
    \end{equation}
    and 
    \begin{equation}\label{EntropiaDiscreta}
        S_k=S_{k-1}+\frac{\ffrdp \parent{q_{k-1},q_k,S_{k-1}}q_{k}-\ffrdm \parent{q_{k-1},q_k,S_{k-1}}q_{k-1}}{\DD_S L_d\parent{q_{k-1},q_k,S_{k-1}}},
    \end{equation}
    for every $k=1,\ldots, N$.
    \end{subequations}
\end{corollary}

\noindent We shall refer to \eqref{eq:EulerLagrangeDiscretas} as the \emph{discrete thermodynamic Euler--Lagrange equations}.

\bigskip

{
We will now incorporate discrete external forces to the constructions above. Let 
$\parent{Q,L_d,\ffrd}$ be a simple discrete thermodynamic system. 
The discrete version of external forces is represented by a $1$-form $\fextd = \left(\fextdm, \fextdp, 0 \right)\in \Omega^1(Q\times Q \times \RR)$.
The $4$-tuple $\parent{Q,L_d,\ffrd, \fextd}$ will be called a \emph{simple discrete forced thermodynamic system}.



\begin{definition}
    A discrete thermodynamic path $\sigma\in \Omega_T\parent{\{t_k\}_{k=0}^N}$ is called a \emph{solution of the discrete thermodynamic forced Euler--Lagrange equations} if for all variations $\delta\sigma\in \mathcal{D}_\sigma$ with fixed extreme points, i.e., $\delta q_0=0=\delta q_N$, we have
    $$
    \dd\mathcal{A}_d\cdot \delta \sigma + \frac{1}{2}\sum_{k=1}^N\Big( \fextdm\left(q_k, q_{k+1}, S_{k}\right) \cdot \delta q_{k} + \fextdp\left(q_{k}, q_{k+1}, S_{k}\right)  \cdot \delta q_{k+1} \Big) = 0 \, .
    $$
\end{definition}

\noindent
To simplify the expressions, we introduce the notation
$$\fd \coloneqq \ffrd + \fextd\, , \quad \fdpm \coloneqq \ffrdpm + \fextdpm\, .$$
Using Theorem~\ref{theorem:discrete_EL_map}, we can obtain the following characterization of the solutions of the discrete thermodynamic forced Euler--Lagrange equations.

\begin{corollary}
    A discrete path $\sigma\in \Omega\parent{\{t_k\}_{k=0}^N}$ is a solution of the discrete thermodynamic forced Euler--Lagrange equations if an only if it satisfies
    \begin{subequations}\label{eq:EulerLagrangeDiscretasforzadas}
    \begin{equation} \label{EulerLagrangeTermoDiscretaforzada}
     \DD_1L_d\parent{q_k,q_{k+1},S_k}+\frac{1}{2}\fdm\parent{q_k,q_{k+1},S_k}+\DD_2L_d\parent{q_{k-1},q_{k},S_{k-1}}+\frac{1}{2}\fdp\parent{q_{k-1},q_{k},S_{k-1}}=0,   
    \end{equation}
    and 
    \begin{equation}\label{EntropiaDiscretaforzada}
        S_k=S_{k-1}+\frac{\ffrdp \parent{q_{k-1},q_k,S_{k-1}}q_{k}-\ffrdm \parent{q_{k-1},q_k,S_{k-1}}q_{k-1}}{\DD_S L_d\parent{q_{k-1},q_k,S_{k-1}}},
    \end{equation}
    for every $k=1,\ldots, N$.
    \end{subequations}
\end{corollary}

\noindent We shall refer to \eqref{eq:EulerLagrangeDiscretasforzadas} as the \emph{discrete thermodynamic forced Euler--Lagrange equations}.

}

\subsection{Discrete flows}\label{sec:discrete_flows}
Now we define and study maps that allow us to update points on $Q\times Q\times \RR$ so that the corresponding discrete path that will be created is a solution of the discrete thermodynamic Euler--Lagrange equations.

{
\begin{definition}\label{def:discrete_Legendre_transform}
    Let $\parent{Q,L_d,\ffrd, \fextd}$ be a simple discrete forced thermodynamic system.
    Assume that $\DD_S L_d$ does not vanish. The \emph{discrete thermodynamic Legendre transforms} are the maps $\mathbb{F}^{f\pm}L_d:Q\times Q\times \RR\longrightarrow \cT Q\times\RR$ given by
    \begin{align*}
    &\mathbb{F}^{f+}L_d\parent{q_0,q_1,S_0}\\
    &=\Bigg(q_1,\DD_2L_d\parent{q_0,q_1,S_0}+\frac{1}{2}\fdp \parent{q_0,q_1,S_0}, 
    S_0+ \frac{\ffrdp \parent{q_{0},q_1,S_{0}}q_{1}-\ffrdm \parent{q_{0},q_1,S_{0}}q_{0}}{\DD_S L_d\parent{q_{0},q_1,S_{0}}}\Bigg),\\ 
    &\mathbb{F}^{f-}L_d\parent{q_0,q_1,S_0}=\parent{q_0,-\DD_1L_d\parent{q_0,q_1,S_0}-\frac{1}{2}\fdm \parent{q_0,q_1,S_0},S_0}\, ,
    \end{align*}
    where $\fdpm \coloneqq \ffrdpm + \fextdpm$.
    The second components of these maps, multiplied by the time step $h$, will be called the discrete momenta, namely
    \begin{equation}\label{eq:discrete_momenta}
    \begin{aligned}
        &p^+_d(q_0,q_1,S_0)=h\DD_2L_d\parent{q_0,q_1,S_0}+\frac{h}{2}\fdp \parent{q_0,q_1,S_0},\\
        & p^-_d(q_0,q_1,S_0)=-h\DD_1L_d\parent{q_0,q_1,S_0}-\frac{h}{2}\fdm \parent{q_0,q_1,S_0}.
    \end{aligned}
    \end{equation}
\end{definition}
}

\begin{remark}
    Notice that $\mathbb{F}^{f-}L_d$ does not depend on the coordinate chart chosen, but, in general, $\mathbb{F}^{f+}L_d$ depends on it, since $\ffrdp \parent{q_{0},q_1,S_{0}}q_{1}$ and $\ffrdm \parent{q_{0},q_1,S_{0}}q_{0}$ depend on the choice of coordinates. As it was previously stated for the equations defining the discrete thermodynamic path space, if $Q$ is a vector space, we choose linear coordinates and $\ffrdp =\ffrdm $, then the definition is coordinate free.
\end{remark}

These discrete thermodynamic Legendre transforms allows us to characterise the solutions of the discrete thermodynamic Euler--Lagrange equations via momentum matching equations as in \cite{M.W2001}. Indeed, definition~\ref{def:discrete_Legendre_transform} and equations \eqref{eq:EulerLagrangeDiscretas} immediatly imply the following:

\begin{proposition}
    A discrete path $\sigma\in \Omega\parent{\{t_k\}_{k=0}^N}$ is a solution of the discrete thermodynamic Euler--Lagrange equations if an only if, for every $k=1,\ldots, N-1$ we have
    $$
    \mathbb{F}^{f+}L_d\parent{q_{k-1},q_k,S_{k-1}}=\mathbb{F}^{f-}L_d\parent{q_k,q_{k+1},S_k}.
    $$
\end{proposition}

\noindent Utilising the inverse function theorem and linear algebra, we can prove the following:

\begin{lemma}
    The discrete thermodynamic Lagrangian transform $\mathbb{F}^{f-}L_d$ is a local diffeomorphism if and only if
    $$
    \DD_2\DD_1L_d+\frac{1}{2}\DD_2\fdm,
    $$
    is regular. In addition, if 
    $$
    \parent{\DD_S L_d}^2+\parent{\DD_S\ffrdp q_2-\DD_S\ffrdm q_1}\DD_SL_D-\parent{\ffrdp q_2-\ffrdm q_1}\DD_S^2L_d\neq 0,
    $$
    then the discrete thermodynamic Lagrangian transform $\mathbb{F}^{f+}L_d$ is a local diffeomorphism if and only if the following matrix is regular:
    {
    \footnotesize{
   \begin{align*}
        & \parent{\DD_1\DD_2L_d+\frac{1}{2}\DD_1\fdp}\parent{\parent{\DD_S L_d}^2+\parent{\DD_S\ffrdp q_2-\DD_S\ffrdm q_1}\DD_SL_D-\parent{\ffrdp q_2-\ffrdm q_1}\DD_S^2L_d}-\\
        & -\parent{\DD_2\DD_SL_d+\frac{1}{2}\DD_S\fdp }\parent{\parent{\parent{\DD_1\ffrdp }q_2-\parent{\DD_1\ffrdm }q_1-\ffrdm }\DD_S L_d-\parent{\ffrdp q_2-\ffrdm q_1}\DD_1\DD_SL_d}\, .
   \end{align*}
    }
    }
\end{lemma}

\begin{definition}
    A discrete Lagrangian $L_d:Q\times Q\times \RR\longrightarrow \RR$ is said to be \emph{semi-regular} if $\mathbb{F}^{f-}L_d$ is a local diffeomorphism. If, in addition, $\mathbb{F}^{f+}L_d$ is a local diffeomophism, $L_d$ will be said to be \emph{regular}. If both discrete thermodynamic Legendre transforms are global diffeomorphisms, $L_d$ will be said to be \emph{hyperregular}.
\end{definition}

\begin{definition}
    Let $L_d:Q\times Q\times \RR\longrightarrow \RR$ be a discrete hyperregular Lagrangian. We define the discrete thermodynamic Lagrangian flow as the map $\Phi_d:Q\times Q\times \RR\longrightarrow Q\times Q\times \RR$ given by
    $$
    \Phi_d=\parent{\mathbb{F}^{f-}L_d}^{-1}\circ \mathbb{F}^{f+}L_d.
    $$
\end{definition}

\begin{remark}
    We can define the discrete thermodynamic Lagrangian flow locally even when the Lagrangian is just semi-regular. In addition, when it is regular, it will be a local diffeomorphism and, if it is hyperregular, a global diffeomorphism.
\end{remark}

\begin{theorem}
    Let $L_d:Q\times Q\times \RR \longrightarrow \RR$ be a discrete hyperregular Lagrangian. Then $\sigma=(q_0,\ldots,q_N,S_0,S_N)\in \Omega\parent{\{t_k\}_{k=0}^N}$ is a solution of the discrete thermodynamic forced Euler--Lagrange equations if and only if
    \begin{equation}\label{EcFlujoLagrangiano}
    \Phi_d\parent{q_{k-1},q_k,S_{k-1}}=\parent{q_k,q_{k+1},S_k},
    \end{equation}
    for all $k=1,\ldots, N-1$.  
\end{theorem}
\begin{proof}
    If $\sigma\in \Omega\parent{\{t_k\}_{k=0}^N}$ is a solution of the discrete thermodynamic forced Euler--Lagrange equations, we know that for every $k=1,\ldots, N-1$, we have $\mathbb{F}^{f-}L_d\parent{q_k,q_{k+1},S_k}=\mathbb{F}^{f+}L_d\parent{q_{k-1},q_k,S_{k-1}}$. Since $\mathbb{F}^{f-}L_d$ is a diffeomorphism, composing with its inverse, we get that $\Phi_d\parent{q_{k-1},q_k,S_{k-1}}=\parent{q_k,q_{k+1},S_k}$.
    \\
    Conversely, if (\ref{EcFlujoLagrangiano}) holds, then composing to the left with $\mathbb{F}^{f-}L_d$ we get that the momentum matching equation holds and, thus, $\sigma$ is a solution of the discrete thermodynamic forced Euler--Lagrange equations. 
\end{proof}

We will now define two discrete $2$-forms that will let us study the evolution of the geometric structure of the problem via the discrete thermodynamic Lagrangian flow.

\begin{definition}
    Let $L_d: Q\times Q\times \RR\longrightarrow \RR$ be a discrete Lagrangian and let $\omega$ be the pullback of the canonical $2$-form on $\cT Q$ to $\cT Q \times \RR$ via the canonical projection. We define
    $$
    \omega^+=\left(\mathbb{F}^{f+}L_d\right)^*\omega, \qquad \omega^-=\left(\mathbb{F}^{f-}L_d\right)^*\omega.
    $$
\end{definition}

\noindent We have that
{
\begin{align*}
    &\left(\mathbb{F}^{f-}L_d\right)^*\dd q=\dd q_{1}\, , \qquad 
    \left(\mathbb{F}^{f+}L_d\right)^*\dd q=\dd q_{2}\, ,\\
    &\left(\mathbb{F}^{f-}L_d\right)^*\dd p=-\parent{\DD_1^2L_d+\frac{1}{2}\DD_1\fdm }\dd q_{1}-\parent{\DD_2\DD_1L_d+\frac{1}{2}\DD_2\fdm }\dd q_{2}\, ,\\
    &\left(\mathbb{F}^{f+}L_d\right)^*\dd p=\parent{\DD_1\DD_2L_d+\frac{1}{2}\DD_1\fdp }\dd q_{1}-\parent{\DD_2^2L_d+\frac{1}{2}\DD_2\fdp}\dd q_{2}\, ,
\end{align*}
and thereupon
$$
\omega^+=\parent{\DD_1\DD_2L_d+\frac{1}{2}\DD_1\fdp }\dd q_{1}\wedge \dd q_{2}\, , \qquad \omega^-=\parent{\DD_1\DD_2L_d+\frac{1}{2}\DD_1\fdm }\dd q_{1}\wedge \dd q_{2}\, .
$$
}
In particular, this implies that $\omega^+ $ is independent of the chosen coordinate chart.

\begin{theorem}
    Let $\parent{Q,L_d,\ffrd, \fextd}$ be a simple discrete forced thermodynamic system. Assume that $L_d$ is hyperregular. Then, 
    $$\Phi_d^*\, \omega^-=\omega^+\, .$$ 
    Moreover, if $\fextd\equiv 0$, then
    $$\Phi_d^*\, \omega^+=\omega^-\, .$$
\end{theorem}
\begin{proof}
    For the first equation, notice that
    $$
    \Phi_d^*\, \omega^-=\Phi_d^*\parent{\left(\mathbb{F}^{f-}L_d\right)^*\omega}=\parent{\mathbb{F}^{f-}L_d\circ \Phi_d}^*\omega=\left(\mathbb{F}^{f+}L_d\right)^* \omega.
    $$
    For the second equation, if $\fextd=0$, fixing $S_0\in \RR$, we define the restricted discrete action map as $\mathcal{A}_d^r:Q\times Q\longrightarrow \RR$ given by
    $$
    \mathcal{A}_d^r\parent{q_0,q_1}=\mathcal{A}_d(\sigma),
    $$
    where $\sigma$ is the only solution of the discrete thermodynamic Euler--Lagrange equations starting at $(q_0,q_1,S_0)$. Thus, for every $v_d=(q_0,q_1)\in Q\times Q$ and $w_{v_d}\in \T_{v_d}(Q\times Q)$ we have, using (\ref{DiferencialAccionTotal}) 
    and the properties of the solutions to the discrete thermodynamic Euler--Lagrange equations, that
    $$
    \dd\mathcal{A}_d^r(v_d)\cdot w_{v_d}=\Theta^+_{L_d}\parent{\Phi_d^{N-1}(v_d)}\cdot \dd\parent{\Phi_d^{N-1}}w_{v_d}-\Theta^-_{L_d}(v_d)\cdot w_{v_d}
    $$
    $$
    =\parent{\parent{\Phi_d^{N-1}}^*\parent{\Theta^+_{L_d}}}(v_d)\cdot w_{v_d}-\Theta^-_{L_d}(v_d)\cdot w_{v_d}.
    $$
    Taking the differential of the last expression and taking into account the expression of the $1$-forms, we get
    $$
    \parent{\parent{\Phi_d^{N-1}}^*\omega^+}(v_d)\cdot w_{v_d}-\omega^-(v_d)\cdot w_{v_d}=0.
    $$
    Since this reasoning is valid for every $N$, taking $N=2$ we obtain the desired result.
\end{proof}

\subsection{Discrete Symmetries and Discrete Noether's theorem}\label{sec:discrete_Noether}

Let $\Phi : G \times Q \longrightarrow Q$ be a Lie group action of the Lie group $G$ on the base manifold $Q$. We define the lifted action $\tilde{\Phi} : G \times Q \times Q \times \RR \to Q \times Q \times \RR$ of $G$ on $Q \times Q \times \RR$ to be the diagonal action on $Q \times Q$ and the identity on
$\RR$, namely,
$$
\tilde{\Phi} (g, q_0, q_1, S) = g (q_0, q_1, S) = \big(\Phi_g (q_0), \Phi_g (q_1), S\big)\, .
$$
Let $\xi_Q\in \X(Q)$ denote the infinitesimal generator on $Q$ associated to a Lie algebra element $\xi \in \mathfrak{g}$, given by
$$\xi_Q (q) = \restr{\frac{\dd }{\dd t}}{t=0} \Phi_{\exp t \xi}(q)\, ,$$
for each $q\in Q$, where $\exp\colon \mathfrak{g}\to G$ is the exponential map. The infinitesimal generator on $Q \times Q \times \RR$ can be identified with the triple
$$
    \xi_{Q\times Q\times\RR} = (\xi_Q, \xi_Q, 0) \in \X(Q\times Q\times \RR)\, ,
$$
using the identification $\T(Q\times Q\times \RR)\cong \T Q \oplus \T Q \oplus \T \RR$.
The \emph{discrete momentum maps} $\mommap_{L_d}^{f+},\, \mommap_{L_d}^{f-}\colon Q\times Q\times \RR\to \mathfrak{g}^*$ are defined by
\begin{equation}
\begin{aligned}
  & \left\langle \mommap_{L_d}^{f+}(q_0, q_1, S_0), \xi   \right\rangle 
  = \left\langle p^+ (q_0,q_,S_0), \xi_Q (q_1) \right\rangle\, ,   \\
  & \left\langle \mommap_{L_d}^{f-}(q_0, q_1), \xi   \right\rangle 
  = \left\langle p^- (q_0,q_1,S_0), \xi_Q (q_0) \right\rangle\, ,
\end{aligned}
\end{equation}
where $p^+$ and $p^-$ are the discrete momenta defined in \eqref{eq:discrete_momenta}. Given an element $\xi$ of $\mathfrak{g}$, if $\left\langle \mommap_{L_d}^{f+}, \xi  \right\rangle = \left\langle \mommap_{L_d}^{f-}, \xi  \right\rangle $, one can define the function
\begin{equation}
    \begin{aligned}
      \mommap_d^\xi\colon Q\times Q\times \RR &\to \RR \\
      (q_0, q_1,S_0) & \mapsto \left\langle \mommap_{L_d}^{f+}, \xi  \right\rangle  (q_0, q_1,S_0) = \left\langle \mommap_{L_d}^{f-}, \xi  \right\rangle  (q_0, q_1,S_0)\, .
    \end{aligned}
\end{equation}  



\begin{theorem}[Discrete forced Noether's theorem]\label{theorem:Noether_discrete}
    The following statements are equivalent:
    \begin{enumerate}
        \item The function $\mommap_d^\xi$ is well-defined, and it is a constant of the motion, namely, $ \mommap_d^\xi  (q_{N-1}, q_N, S_{N-1}) = \mommap_d^\xi (q_0, q_1, S_0)$ along the discrete Lagrangian flow.
        \item $\xi_{Q\times Q\times\RR}(L_d) + \fd(\xi_{Q\times Q\times\RR}) = 0$.
    \end{enumerate} 
\end{theorem}

\noindent The proof, \textit{mutatis mutandis}, is identical to the one from \cite[Theorem~2]{d.L.L2022a}.

\begin{table}[H]
    $$\begin{array}{|c|c|c|}
        \hline
        & \text{Our discretization} & \text{Gay-Balmaz and Yoshimura}  \\
        \hline
        \text{Space} &Q\times Q \times \RR & Q\times Q \times \RR \times \RR\\
        \hline
        \text{Evolution} & 
        \text{\tiny $
        \begin{array}{l}
             \DD_1L_d\parent{q_k,q_{k+1},S_k}+\frac{1}{2}\fdm\parent{q_k,q_{k+1},S_k}\\
             +\DD_2L_d\parent{q_{k-1},q_{k},S_{k-1}}+\frac{1}{2}\fdp\parent{q_{k-1},q_{k},S_{k-1}}=0\, ,\\\\
             S_k=S_{k-1}+\frac{\ffrdp \parent{q_{k-1},q_k,S_{k-1}}q_{k}-\ffrdm \parent{q_{k-1},q_k,S_{k-1}}q_{k-1}}{\DD_S L_d\parent{q_{k-1},q_k,S_{k-1}}}\, .
        \end{array}$}
        & 
        \text{\tiny $
        \begin{array}{l}
             \DD_1L_d\parent{q_k,q_{k+1},S_k,S_{k+1}}\\
             +\frac{1}{2}\fdm\parent{q_k,q_{k+1},S_k,S_{k+1}}\\
             +\DD_2L_d\parent{q_{k-1},q_{k},S_{k-1},S_k}
             \\+\frac{1}{2}\fdp\parent{q_{k-1},q_{k},S_{k-1},S_k}=0\, ,\\\\
             \parent{q_k,q_{k+1},S_k,S_{k+1}}\in C_d\, .
        \end{array}$}\\
        \hline
    \end{array}
    $$
    \caption{Comparison between the discretisation provided by our method and the one developed in \cite{G.Y2018a} for a Lagrangian thermodynamic system $(\T Q\times \RR, \omega_L, \eta_L, L, \Ftext, \Ftfr)$. Note that in their approach the evolution is given by an equation and a certain constraint $C_d\subseteq Q\times Q \times \RR \times \RR$.}
    \label{table:comparison_approaches}
\end{table}

\section{Examples}\label{sec:examples}
\subsection{Damped harmonic oscillator}\label{subsec:example_harmonic_oscillator}
Let us consider a one dimensional damped harmonic oscillator, such as in \cite{A.d.L+2020}, whose continuous Lagrangian $L: \T\RR\times \RR\cong \RR^3\to \RR$ is given by
$$
L(q,v,S)=\frac{1}{2}v^2-\frac{1}{2}q^2-\gamma S\, .
$$
We will consider Rayleigh type friction forces with a friction coefficient proportional to the temperature ($\gamma=T=-\pder{L}{S}$), given by
$$
\Ffr(q,v)=-\gamma v \dd q\, .
$$
{ We will consider no external forces act on the system. }The corresponding thermodynamic Euler--Lagrange equations \eqref{eq:thermodynamic_Euler_Lagrange} read 
\begin{equation}\label{eq:ExContinuous}
    \Ddot{q}+\gamma \dot{q} +q=0\, , \qquad \dot{S}=\dot{q}^2\, .
\end{equation}
In particular, note that the first equation is the evolution of a damped harmonic oscillator.
Notice that, as explained in Remark \ref{remark:RelationToContact},  equations \eqref{eq:ExContinuous} are the thermodynamic Herglotz equations for the Lagrangian $L$. The general solution of these equations is given by
$$
q(t)=Ae^{-\frac{\gamma t}{2}}\cos\parent{t\, \sqrt{1-\parent{\frac{\gamma}{2}}^2}+\phi}\, ,
$$
where $A$ and $\phi$ are arbitrary real constants. Hereinafter, we will consider the solution with initial position $q=0$ and initial velocity $v=1$.

Given a fixed time step $h\in \RR$, we will approximate the continuous Lagrangian $L$ by the discrete Lagrangian $L_d:Q\times Q\times \RR \to \RR$ given by
$$
L_d(q_0,q_1,S_0)=\frac{\parent{q_1-q_0}^2}{2h^2}- \frac{\parent{q_1+q_0}^2}{8}-\gamma S_0\, ,
$$
as prescribed by the midpoint rule substitutions:
$$q\mapsto \frac{q_1+q_0}{2}\, , \quad v\mapsto \frac{q_1-q_0}{h}\, .$$
Similarly, we will approximate the friction force $\Ffr$ by the discrete friction forces given by
$$
\ffrdm\parent{q_0,q_1,S_0}=-\gamma \frac{q_1-q_0}{h} \dd q_0\, , \quad
\ffrdp\parent{q_0,q_1,S_0}=-\gamma \frac{q_1-q_0}{h} \dd q_1\, ,
$$
{ and take $\ffrdpm=0$. }The discrete thermodynamic Euler--Lagrange equations \eqref{eq:EulerLagrangeDiscretas} for this system provide the following integration method:
$$ q_{k+1}=\frac{2\parent{4-h^2}}{4+h(h+2\gamma)}q_{k}-\frac{4+h^2-4h\gamma}{4+h(h+2\gamma)}q_{k-1}\, , \quad
S_k=S_{k-1}+\frac{\parent{q_k-q_{k-1}}^2}{h}\, .$$
For the simulation, we have considered a factor $\gamma=0.1$.
Using a decreasing step size, we have carried out the integration of equations (\ref{eq:ExContinuous}) from $t=0$ to $t=1000$, using the discrete thermodynamic Euler--Lagrange equations, with the exact values as initial data for the first two points of the discrete solution. We have used an implementation of the Runge--Kutta~2 method of the midpoint rule for comparison of the results (see Table~\ref{table:position_oscillator} and Figure~\ref{fig:position_oscillator}). We have also made a comparison between the value of the entropy given by the continuous method and the values provided by both the variational and midpoint rule integrators  (see Table~\ref{table:entropy_oscillator} and Figure~\ref{fig:entropy_oscillator}). Since the method used for the calculation of the continuous entropy produces some outliers, we have compared only the first 1500 values of each integration.

\noindent Via the Legendre transform, we can define the Hamiltonian counterpart of the (continuous) Lagrangian $L$:
$$
H(q,p,S)= \frac{1}{2}p^2+\frac{1}{2}q^2+\gamma S,
$$
which is a conserved quantity. In the discrete setting, we can estimate the Hamiltonian using the variational integrator in three different ways: evaluating it in either of the two Legendre transform maps, $\mathbb{F}^{f+}, \mathbb{F}^{f-}$, or by using the straightforward estimation of the velocity from the integrated curve, $v\sim \frac{q_k-q_{k-1}}{h}$, and that in the continuous setting for this example, $v=p$. The first two methods will yield a value of the Hamiltonian almost constant (up to $10^{-12}$) and nearly identical (up to $10^{-14}$), while the latter will produce an oscillating value with a higher error.

For obtaining the value of the Hamiltonian using the midpoint rule, we again note that, in the continuous setting for this example $v=p$, so that we can use the values of $v$ provided by this method to estimate the Hamiltonian. The resulting value has a greater error, and it decays as $t$ increases initially, as can be seen in Table~\ref{table:Hamiltonian_oscillator} and Figure~\ref{fig:Hamiltonian-Harmonic}.

\begin{figure}[H]
    \centering
    \begin{subfigure}{.495\linewidth}
        \centering
        \includegraphics[width=\linewidth]{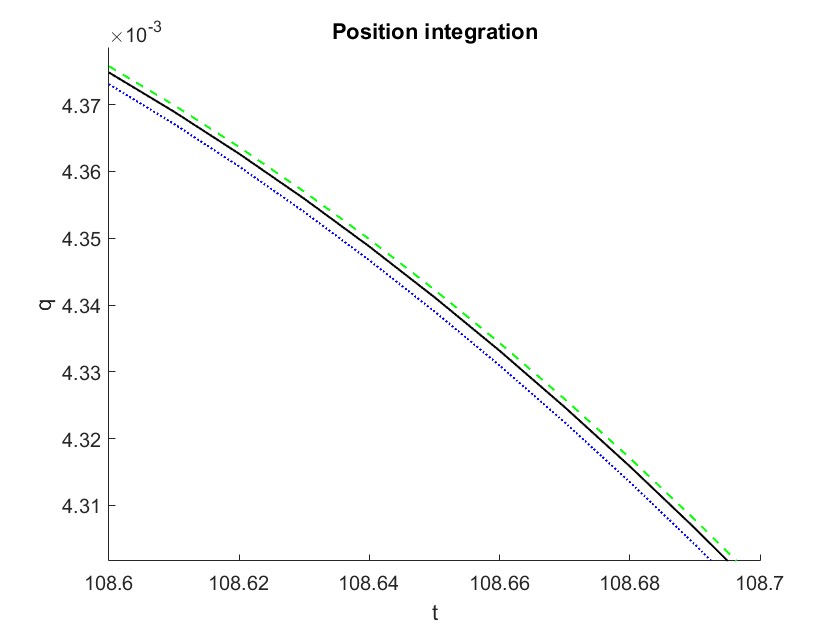}
        \caption{Position}
        \label{fig:position_oscillator}
    \end{subfigure}
    \begin{subfigure}{.495\linewidth}
        \centering
        \includegraphics[width=\linewidth]{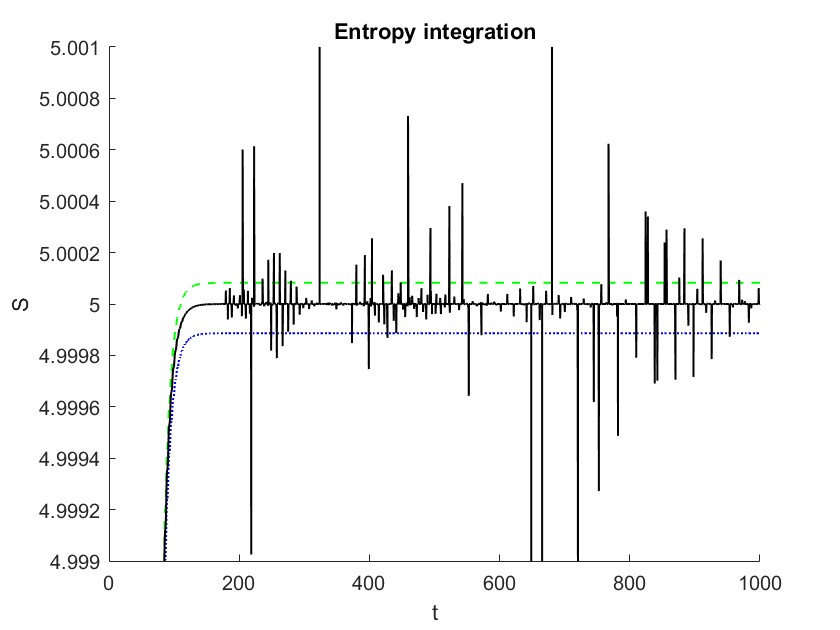}
        \caption{Entropy}
        \label{fig:entropy_oscillator}
    \end{subfigure}
    
    \caption{Integration results for the damped harmonic oscillator with a time step $h=0.01$ using the Discrete Thermodynamic Euler--Lagrange equations solution (in green and dashed) and the method of the midpoint rule (in blue and dotted), compared with the exact continuous solution (in black and solid). A detailed section is shown to allow the distinction of the curves, and to show the numeric errors made in the calculation of the continuous solution for $S$.
    }
\end{figure}

\begin{figure}[H]
    \centering
    \includegraphics[scale=0.35]{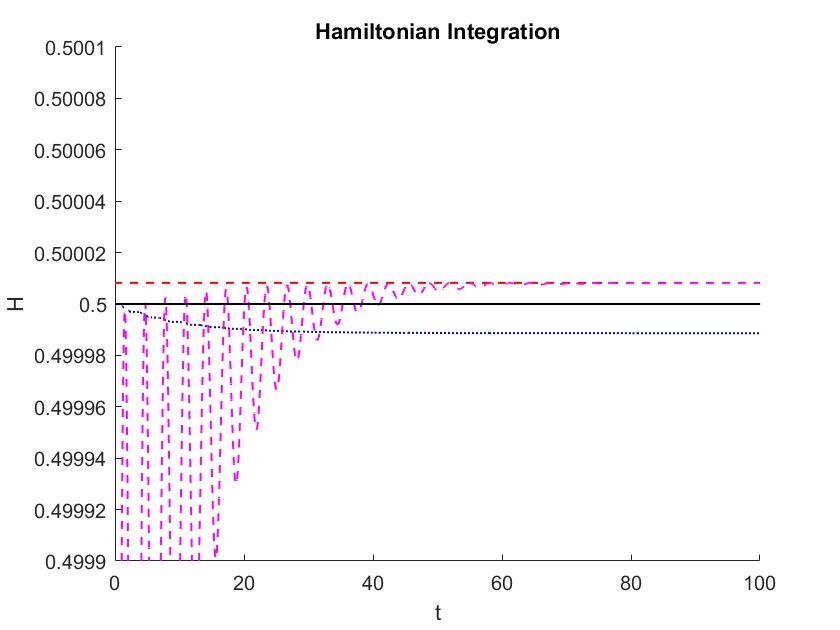}
    \caption{\label{fig:Hamiltonian-Harmonic} Estimation of the value of the Hamiltonian for the damped harmonic oscillator with a time step $h=0.01$ using the Discrete Thermodynamic Euler--Lagrange equations solution and the Legendre transforms $\mathbb{F}^{f+}$ (in green and dashed), $\mathbb{F}^{f-}$ (in red and dashed), as well as the estimation for the velocity (in magenta and dashed) and the method of the midpoint rule (in blue and dotted), compared with the exact continuous solution (in black and solid). The green and red curves are overlapping so they cannot be distinguished.}
    \label{fig:Hamiltonian_oscillator}
\end{figure}

\begin{table}[H]
    \centering
    \begin{subtable}[t]{\linewidth}
        \centering
        \begin{tabular}{c c c}
        \hline 
            Step Size &Discrete thermodynamic E-L equations & Midpoint Method \\
            \hline\hline
            $10^{-1}$ & $6.180 \times 10^{-3}$ & $1.228 \times 10^{-2}$ \\
            $10^{-2}$ &$6.182 \times 10^{-5}$ & $1.226 \times 10^{-4}$ \\
            $10^{-3}$ &$6.173 \times 10^{-7}$ & $1.225 \times 10^{-6}$ \\
            $10^{-4}$ & $2.860 \times 10^{-8}$ & $1.225 \times 10^{-8}$ \\
            \hline
        \end{tabular}
        \caption{Position}
        \label{table:position_oscillator}
    \end{subtable}
    
    \bigskip
    
    \begin{subtable}[t]{\linewidth}
        \centering
        \begin{tabular}{c c c}
            \hline 
            Step Size &Discrete thermodynamic E-L equations & Midpoint Method \\
            \hline\hline
            $10^{-1}$ & $8.10 \times 10^{-3}$ & $7.63 \times 10^{-3}$ \\
            $10^{-2}$ &$3.36 \times 10^{-5}$ & $5.41 \times 10^{-4}$ \\
            $10^{-3}$ &$2.02 \times 10^{-8}$ & $1.26 \times 10^{-7}$ \\
            $10^{-4}$ & $1.49 \times 10^{-11}$ & $4.65 \times 10^{-10}$ \\
            \hline
        \end{tabular}
        \caption{Entropy}
        \label{table:entropy_oscillator}
    \end{subtable}

    \bigskip

    \begin{subtable}[t]{\linewidth}
        \centering
        \begin{tabular}{c c c c}
        \hline 
          Method $p^+_d$ &   Method $p^-_d$ & Variational Method ($v$) & Midpoint Method ($v$) \\
        \hline\hline
         $8.10 \times 10^{-4}$ & $8.10 \times 10^{-4}$ & $3.80 \times 10^{-3}$ & $5.62 \times 10^{-5}$ \\
        $8.24 \times 10^{-6}$ & $8.24 \times 10^{-6}$ & $4.91 \times 10^{-4}$ & $1.14 \times 10^{-5}$ \\
        $8.25 \times 10^{-8}$ & $8.25 \times 10^{-8}$ & $4.99 \times 10^{-5}$ & $1.25 \times 10^{-7}$ \\
        $5.44 \times 10^{-9}$ & $5.44 \times 10^{-9}$ & $5.00 \times 10^{-6}$ & $1.26 \times 10^{-9}$ \\
        \hline
        \end{tabular}
        \caption{Hamiltonian, calculated for $h=0.1$, $h=0.01$, $h=0.0001$ and $h=0.00001$.}
        \label{table:Hamiltonian_oscillator}
    \end{subtable}
    
    \caption{Greatest absolute difference of the exact solution of the damped harmonic oscillator and two numerical approximations: the solution of the discrete thermodynamic Euler--Lagrange equations, and the result of RK-2 method of the midpoint rule.}
\end{table}


\subsection{Ideal gas in a piston}\label{subsec:example_ideal_gas}

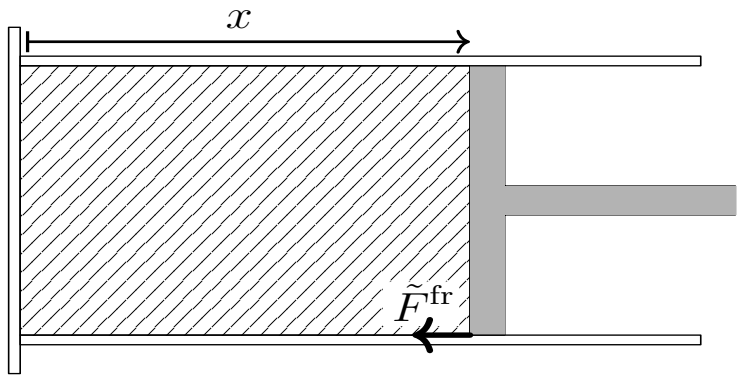
\begin{figure}[h]
\centering
\resizebox{10 cm}{5cm}{
\begin{tikzpicture}[line cap=round, line join=round, rotate=270]
    \def\R{1.4}      
    \def\H{6}        
    \def\P{4}        
    \def\T{0.3}        
    \def\RodH{2}     
    \def\RodW{0.15}
    \def\WallsWidth{0.1}
    \def\GroundHeight{0.1}
    \def\BaselineWidth{1.8}
    \def\Separation{0.15}
    \def\RozHeight{0.5}
    \def\BlankWidth{0.55}
    \def\BlankHeight{0.75}

    \draw[pattern={Lines[angle=45, distance=5pt, line width=0.01pt]}] (-\R,\GroundHeight) rectangle (\R,\P);
    \draw [fill=white, white] (\R-\BlankWidth,\P-\BlankHeight) rectangle (\R-0.1,\P-0.1);
    \draw [thin] (-\R,\P+\T)--(-\RodW,\P+\T)--(-\RodW,\P+\T+\RodH)--(\RodW,\P+\T+\RodH)--(\RodW,\P+\T)--(\R,\P+\T)--(\R,\P)--(-\R,\P)--(-\R,\P+\T);
    \fill[gray!60] (-\R,\P) rectangle (\R,\P+\T);
    \fill[gray!60] (-\RodW,\P+\T-0.1) rectangle (\RodW,\P+\T+\RodH);
    \draw (-\R-\WallsWidth,\GroundHeight) rectangle (-\R,\H);
    \draw ( \R,\GroundHeight) rectangle ( \R+\WallsWidth,\H);
    \draw ( -\BaselineWidth,0) rectangle (\BaselineWidth,\GroundHeight);
    \draw [thick,|->] (-\R-\WallsWidth-\Separation,\GroundHeight+0.05)--(-\R-\WallsWidth-\Separation,\P);
    \node [sloped, above] at ((-\R-\WallsWidth-\Separation,\P/2){$x$};
    \draw [ultra thick,<-] (\R,\P-\RozHeight)--(\R,\P);
    \node [sloped, above] at ((\R,\P-\RozHeight*0.8){$\displaystyle \Ftfr$};
    \end{tikzpicture}
    }
    \caption{Diagram of the cylinder containing the gas (dashed region) and the piston that closes it.}
\end{figure}
    Consider a perfect gas confined by a piston and contained in a cylinder, as was done by Gay-Balmaz and Yoshimura in \cite{Balmaz+2017}. In suitable units and for a suitable number of particles in the gas, the Lagrangian that describes this system is given by
    \begin{equation*}
        L(x, v_x ,S)=\frac{v_x^2}{2}-e^Sx^{-\frac{1}{c}},
    \end{equation*}
    where $c=\frac{3}{2}$ is the molar specific heat capacity at constant volume divided by $R$, the constant of ideal gasses, for a monoatomic gas.

 We also consider a friction force between the cylinder and the piston, which is proportional to its velocity, {and no external force acting on the piston.} In this case, we consider a constant temperature independent coefficient of friction $\gamma$:
    \begin{equation*}
        \Ftfr=-\gamma v_x \dd x.
    \end{equation*}
    The thermodynamic Euler--Lagrange equations for this system are given by
    \begin{equation}\label{eq:PistonContinuous}
        \Ddot{x}+\gamma \dot{x}-\frac{1}{c}e^Sx^{-\parent{1+\frac{1}{c}}}=0, \qquad \dot{S}=\gamma \dot{x}^2x^{\frac{1}{c}}e^{-S}.
    \end{equation}
    For a fixed time step $h\in \RR$, we will approximate the continuous Lagrangian by the discrete Lagrangian $L_d:Q\times Q\times\RR\to \RR$ given by
    \begin{equation*}
        L_d\parent{q_0,q_1,S_0}=\frac{\parent{q_1-q_0}^2}{2h^2}-e^{S_0}\parent{\frac{q_1+q_0}{2}}^{-\frac{1}{c}},
    \end{equation*}
    {the external forces $\fextdpm=0$} and the friction forces by
    \begin{equation*}
        \ffrdm\parent{q_0,q_1,S_0} = -\gamma\frac{q_1-q_0}{h}\, \dd q0\, , \quad \ffrdp\parent{q_0,q_1,S_0}=-\gamma\frac{q_1-q_0}{h}\, \dd q_1.
    \end{equation*}
    These are obtained in the same way as in the previous example. The corresponding discrete thermodynamic Euler--Lagrange equations give the following implicit integration method:
    \begin{align*}
        &-q_{k+1}\parent{\frac{1}{h^2}+\frac{\gamma}{2h}}+\frac{2q_k}{h^2}+q_{k-1}\parent{\frac{\gamma}{2h}-\frac{1}{h^2}}+\\
        &+\frac{1}{2c}\parent{e^{S_k}\parent{\frac{q_{k+1}+q_k}{2}}^{-\parent{1+\frac{1}{c}}}+e^{S_{k-1}}\parent{\frac{q_{k}+q_{k-1}}{2}}^{-\parent{1+\frac{1}{c}}}}=0, \\
        &S_k=S_{k-1}+\frac{\gamma}{h}\parent{q_k-q_{k-1}}^2\parent{\frac{q_k+q_{k-1}}{2}}^{\frac{1}{c2}}e^{-S_{k-1}}.
    \end{align*}
    
    Using a step size $h=0.01$ and a friction coefficient $\gamma=0.1$, we have carried out the integration of equations (\ref{eq:PistonContinuous}) with the variational method and the midpoint rule for initial conditions $x_0=1$, $x_1=1$ (i.e., initial velocity $v_x=0$) and $S_0=10$. Both methods have been compared with the integration obtained by the higher order Runge--Kutta method provided by the function \textit{ode45} of Matlab, using both relative and absolute tolerances of $10^{-10}$ and interpolation between the points resulting from this integration. See Table~\ref{table:position_entropy_ideal_gas} and Figure~\ref{fig:position_ideal_gas}. Moreover, a comparison of the estimations of the Hamiltonian with its initial value (constant in the continuous model) has been carried out, following the same outline as in the previous example (see Table~\ref{table:Hamiltonian_ideal_gas} and Figure~\ref{fig:Hamiltonian_ideal_gas}).

\begin{figure}[H]
    \centering
    \includegraphics[scale=0.35]{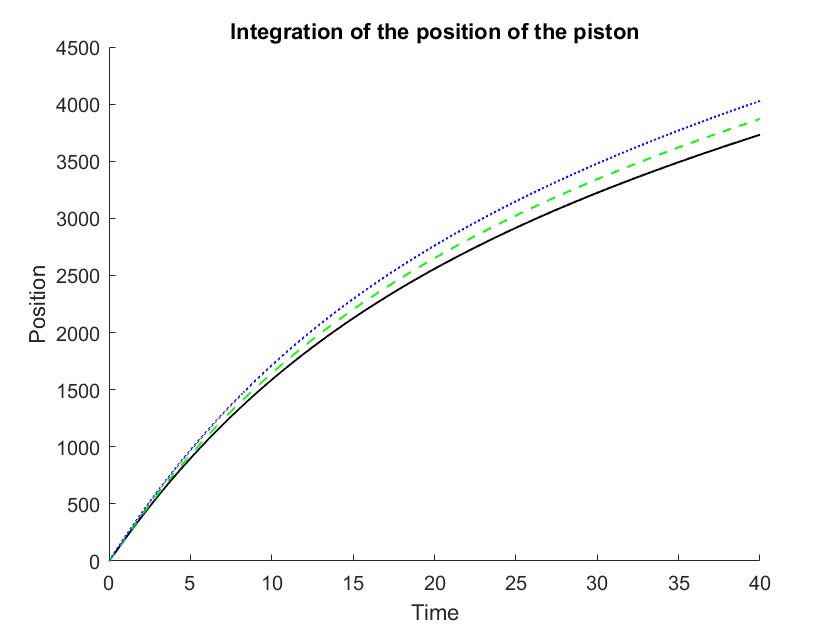}
    \caption{Integration of $x$ for a perfect monoatomic gas contained in a cylinder with a time step $h=0.01$ using the Discrete thermodynamic Euler--Lagrange equations solution (in green and dashed) and the method of the midpoint rule (in blue and dotted), compared with the exact continuous solution (in black and solid).}
 \label{fig:position_ideal_gas}
\end{figure}

\begin{figure}[H]
    \centering
    \includegraphics[scale=0.35]{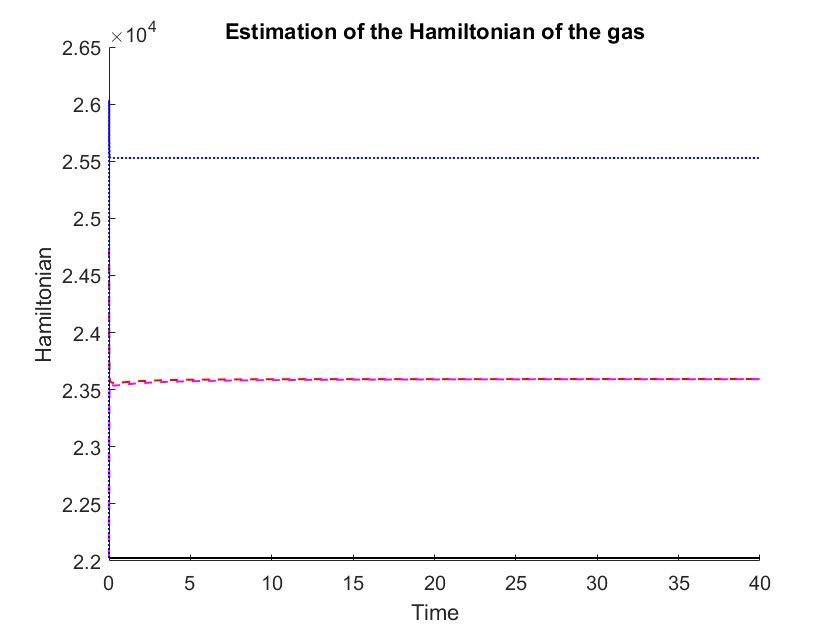}
    \caption{Estimation of the value of the Hamiltonian for a perfect monoatomic gas contained in a cylinder with a time step $h=0.01$ using the Discrete thermodynamic Euler--Lagrange equations solution and the Legendre transforms $\mathbb{F}^{f+}$ (in green and dashed), $\mathbb{F}^{f-}$ (in red and dashed), as well as the estimation for the velocity (in magenta and dashed) and the method of the midpoint rule (in blue and dotted), compared with the exact continuous solution (in black and solid). The green and red curves are overlapping so they cannot be distinguished.}
    \label{fig:Hamiltonian_ideal_gas}
\end{figure}

\begin{table}[H]
    \centering
    \begin{subtable}[t]{\linewidth}
        \centering
        \begin{tabular}{c c c}
        \hline 
         &Discrete thermodynamic E-L equations & Midpoint Method \\
        \hline\hline
         Position & $55.7 $ & $ 125.2$ \\
         Entropy & $9.42 \times 10^{-2}$ & $ 2.01 \times 10^{-1}$ \\
        \hline
        \end{tabular}
        \caption{Position and entropy}
        \label{table:position_entropy_ideal_gas}
    \end{subtable}

    \bigskip
    
    \begin{subtable}[t]{\linewidth}
        \centering
        \begin{tabular}{c c c c}
        \hline 
          Method $p^+_d$ &   Method $p^-_d$ & Variational Method 
          ($v$) & Midpoint Method ($v$) \\
        \hline\hline
         $1.89\times 10^3 $ & $2.70 \times 10^{3}$ & $1.56 \times 10^{3}$ & $4.01 \times 10^{3}$ \\
        \hline
        \end{tabular}
        \caption{Hamiltonian}
        \label{table:Hamiltonian_ideal_gas}
    \end{subtable}
    \caption{Greatest absolute difference between the exact solution and the numerical solutions obtained using the Discrete thermodynamic Euler--Lagrange equations solution and the result of RK-2 method of the midpoint rule, for the evolution of a perfect monoatomic gas contained in a cylinder.}
\end{table}


\subsection{Van der Waals gas in a piston}\label{subsec:example_Van_der_Waals}
    A more realistic example would be to consider a Van der Waals gas formed by particles without internal degrees of freedom, contained in a cylinder similar to that in the previous example. Considering a suitable value for the mass of the cylinder, as well as for the number of particles in the gas, the Lagrangian in suitable units is given by \cite{Johnston2014}:
    \begin{equation*}
        L\parent{x,v_x,S}=\frac{v_x^2}{2}-\parent{\frac{1}{x-\hat{b}}}^{2/3}e^S+\frac{\hat{a}}{x},
    \end{equation*}
    where $\hat{a},\hat{b}$ are constants proportional to the usual parameters $a,b$ of Van der Waals model, representing the average value of the potential energy per unit of concentration for the intermolecular potential and the volume of a molecule, respectively.

    Assuming that the velocity of the piston will be small, it is reasonable to again assume that the friction force between the cylinder and the piston is proportional to the velocity,with the proportionality constant being $\gamma$:
    \begin{equation*}
    \Ftfr=-\gamma v_x \dd x.
    \end{equation*}
    { We will assume again no external forces act on the piston. }In this setting, the thermodynamic Euler--Lagrange equations are the following:
    \begin{equation}\label{eq:EulerLagrangeVanDerWaals}
        \Ddot{x}+\gamma \dot{x}-\frac{2}{3}\parent{\frac{1}{x-\hat{b}}}^{5/3}e^S+\frac{a}{x^2}=0, \qquad \dot{S}=\gamma \dot{x}^2\parent{x-\hat{b}}^{2/3}e^{-S}.
    \end{equation}
    Given a fixed time step $h$, we will approximate the continuous Lagrangian by the discrete version obtained using the same method as in the previous examples:
    \begin{equation*}
        L_d\parent{q_0,q_1,S_0}=\frac{\parent{q_1-q_0}^2}{2h^2}-\parent{\frac{2}{q_1+q_0-2b}}^{2/3}e^{S_0}+\frac{2a}{q_1+q_0},
    \end{equation*}
    {the external forces $\fextdpm=0$} and the friction forces as follows: 
    \begin{equation*}
    \ffrdm\parent{q_0,q_1,S_0}=-\gamma \frac{q_1-q_0}{h}\, \dd q_0\, , \quad \ffrdp\parent{q_0,q_1,S_0}=-\gamma \frac{q_1-q_0}{h}\, \dd q_1.
    \end{equation*}
    The corresponding discrete thermodynamic Euler--Lagrange equations for the system give the following implicit integration method:
    \begin{align*}
        &\parent{-\frac{1}{h^2}-\frac{\gamma}{2h}}q_{k+1}+\frac{2q_k}{h^2}+\parent{\frac{\gamma}{2h}-\frac{1}{h^2}}q_{k-1}+\frac{e^{S_k}}{3}\parent{\frac{2}{q_{k+1}+q_k-2b}}^{5/3}+\\
        &+\frac{e^{S_{k-1}}}{3}\parent{\frac{2}{q_{k}+q_{k-1}-2b}}^{5/3}-2a\parent{\frac{1}{\parent{q_{k+1}+q_k}^2}+\frac{1}{\parent{q_{k}+q_{k-1}}^2}}=0, \\
        &S_k=S_{k-1}+\frac{\gamma}{h}e^{-S_{k-1}}\parent{q_k-q_{k-1}}^2\parent{\frac{q_k+q_{k-1}-2b}{2}}^{2/3}.
    \end{align*}
    We have carried out the integration of equations (\ref{eq:EulerLagrangeVanDerWaals}) using a time step $h=0.01$, taking $\gamma=0.1$, $a=10^3$, $b=0.1$ and for initial conditions $q_0=1$, $S_0=10$ and $v_0=0$ and, thus $q_1=q_0$. As in the previous example, we have compared both methods with the integration provided by the function \textit{ode45} of Matlab, using relative and absolute tolerances of $10^{-10}$, and interpolation between the points resulting from this interpolation. We also present a comparison of the estimation of the Hamiltonian with its initial value, analogous to those in the previous examples. See Table~\ref{table:van_der_Waals}, as well as Figures~\ref{fig:position_van_der_Waals} and~\ref{fig:Hamiltonian_van_der_Waals}.

\begin{table}[H]
    \centering
    \begin{subtable}[t]{\linewidth}
        \centering
        \begin{tabular}{c c c}
        \hline 
         &Discrete thermodynamic E-L equations & Midpoint Method \\
        \hline\hline
         Position & $67.3 $ & $ 178.1$ \\
         Entropy & $1.11 \times 10^{-1}$ & $ 2.78 \times 10^{-1}$ \\
        \hline
        \end{tabular}
        \caption{Position and entropy}
    \end{subtable}

    \begin{subtable}[t]{\linewidth}
        \centering
        \begin{tabular}{c c c c}
        \hline 
          Method $p^+_d$ &   Method $p^-_d$ & Variational Method 
          ($v$) & Midpoint Method ($v$) \\
        \hline\hline
         $2.28\times 10^3 $ & $3.40 \times 10^{3}$ & $1.93 \times 10^{3}$ & $5.74 \times 10^{3}$ \\
        \hline
        \end{tabular}
        \caption{Hamiltonian}
    \end{subtable}
    \caption{Greatest absolute difference between the exact solution and the numerical solutions obtained using the Discrete thermodynamic Euler--Lagrange equations solution and the result of RK-2 method of the midpoint rule, for the evolution of a Van der Waals gas contained in a cylinder.}
    \label{table:van_der_Waals}
\end{table}


\begin{figure}[H]
    \centering
    \includegraphics[scale=0.35]{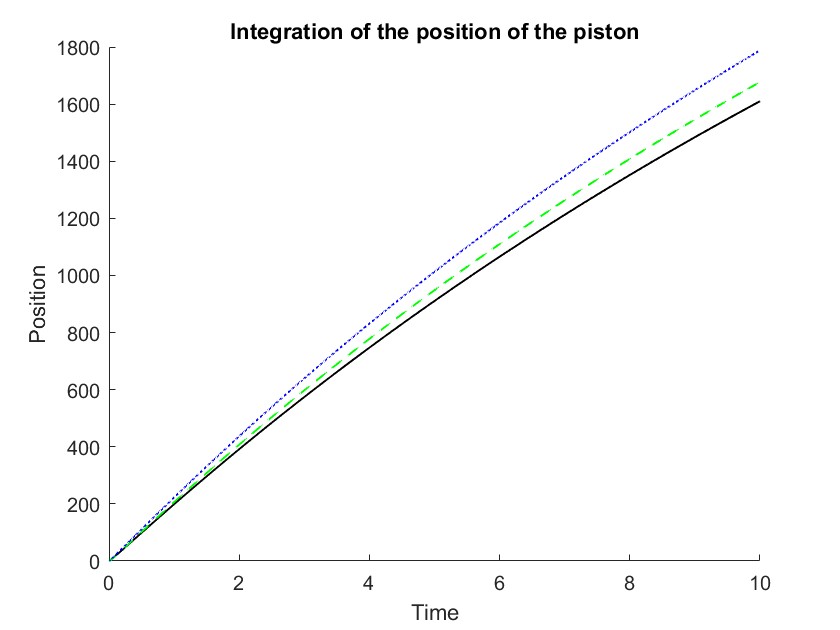}
    \caption{Integration of $x$ for a Van der Waals gas composed of molecules without internal degrees of freedom contained in a cylinder with a time step $h=0.01$ using the Discrete thermodynamic Euler--Lagrange equations solution (in green and dashed) and the method of the midpoint rule (in blue and dotted), compared with the exact continuous solution (in black and solid).}
    \label{fig:position_van_der_Waals}
\end{figure}

\begin{figure}[H]
    \centering
    \includegraphics[scale=0.35]{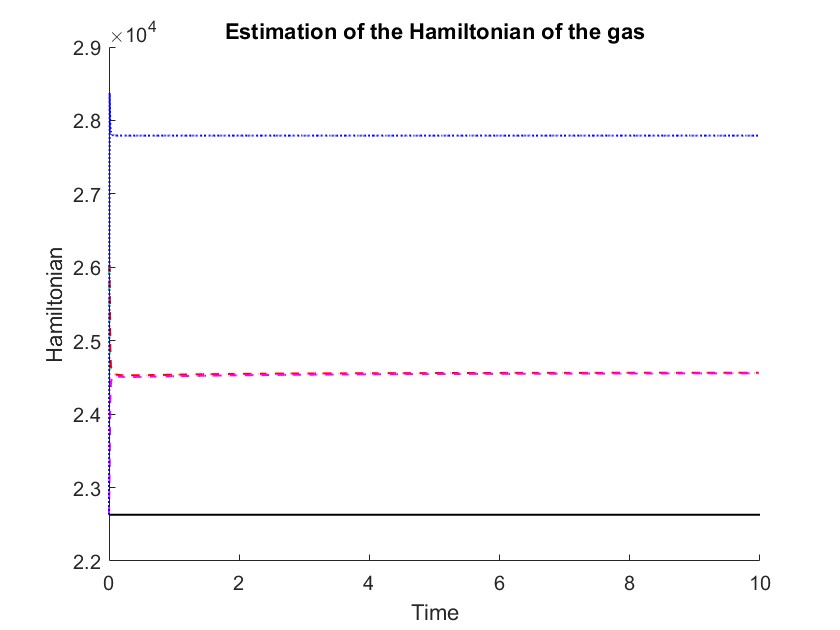}
    \caption{\label{fig:Hamiltonian-VanderWaals} Estimation of the value of the Hamiltonian for a Van der Waals gas composed of molecules without internal degrees of freedom contained in a cylinder with a time step $h=0.01$ using the Discrete thermodynamic Euler--Lagrange equations solution and the Legendre transforms $\mathbb{F}^{f+}$ (in green and dashed), $\mathbb{F}^{f-}$ (in red and dashed), as well as the estimation for the velocity (in magenta and dashed) and the method of the midpoint rule (in blue and dotted), compared with the exact continuous solution (in black and solid). The green and red curves are overlapping so they cannot be distinguished.}
    \label{fig:Hamiltonian_van_der_Waals}
\end{figure}

\subsection{Expansion of a Van der Waals gas against external pressure}\label{subsec:example_Van_der_Waalsb}

    Most of the physically interesting examples of thermodynamic systems involve a non-zero external force. For the case of a Van der Waals gas contained in a piston, this is usually given by the external pressure.

    To make a comparison of the results given by the discrete thermodynamic forced Euler-Lagrange equations and the ones obtained via the GENERIC-Motivated method introduced in  \cite{Fulop2025}, we will work with the reduced version of the Van der Waals equations. The continuous Lagrangian for this system, considering a monoatomic gas with three degrees of freedom, is given by:
    \begin{equation*}
        L\parent{V_r,B_r,S_r}=m_p\frac{B_r^2}{2A^2}-\parent{\frac{2^{4}}{3V_r-1}}^{2/3}\mathrm{exp}\parent{\parent{S_r-1}\parent{\frac{2}{3}\ln{2x_c}+\frac{5}{3}}}+\frac{3}{V_r},
    \end{equation*}
    where $V_r,B_r,S_r$ are the reduced versions of the volume enclosed by the piston, the change rate of that volume in time and the entropy, $A$ is the area of the piston, $m_p$ its mass and $x_c$ is an adimensional coefficient that depends on the parameters of the gas being modeled, which is given by:
    \begin{equation*}
        x_c=\frac{T_c}{8p_c}\parent{\frac{m_{\mathrm{gas}}T_c}{2\pi\hbar}}^{\frac{2}{3}}\, ,
    \end{equation*}
    where $m_{\mathrm{gas}}$ is the mass of a gas molecule. We have assumed that the gas in the piston is He-4 and used the values of its critical magnitudes and molar mass in \cite{He41989}. From now on, we will consider, as in \cite{Fulop2025}, the quotient $\frac{m_p}{A^2}=1$.
    
    Assuming that the velocity of the piston will be small, it is reasonable to again assume that the friction force between the cylinder and the piston is proportional to the velocity,with the proportionality constant being $\gamma$:
    \begin{equation*}
    \Ftfr=-\gamma B_r \dd V_r\, .
    \end{equation*}
    We will also consider an external force given by the reduced pressure of the exterior:
    \begin{equation*}
    \Ftfr=-p_A \dd V_r\, .
    \end{equation*}
    Given a fixed time step $h$, we will approximate the continuous Lagrangian by the discrete version obtained using the same method as in the previous examples:
    \begin{equation*}
        L_d\parent{q_0,q_1,S_0}=\frac{\parent{q_1-q_0}^2}{2h^2}-\parent{\frac{2^{5}}{3\parent{q_1+q_0}-2}}^{2/3}\mathrm{exp}\parent{\parent{S_0-1}\parent{\frac{2}{3}\ln{2x_c}+\frac{5}{3}}}+\frac{6}{q_1+q_0}\, ,
    \end{equation*}
    and the friction and external forces as follows: 
    \begin{align*}
    &\ffrdm\parent{q_0,q_1,S_0}=-\gamma \frac{q_1-q_0}{h}\, \dd q_0\, ,  &\ffrdp\parent{q_0,q_1,S_0}=-\gamma \frac{q_1-q_0}{h}\, \dd q_1\, ,\\
    &\fextdm\parent{q_0,q_1,S_0}=-p_A\, \dd q_0\, ,  &\fextdp\parent{q_0,q_1,S_0}=-p_A\, \dd q_1\, . 
    \end{align*}

    We have carried out the integration of the discrete thermodynamic forced Euler-Lagrange equations for this model using a time step $h=0.01$, taking $\gamma=0.1$, $p_A=1.8$ and for initial conditions $q_0=1.5$, $S_0=1.003702793854103$ and $B_r(0)=0$ and, thus $q_1=q_0$, which correspond to the initial conditions $V_r=1.5$, $T_r=1.01$.

    The results have been compared with the ones given by the GENERIC-motivated extended symplectic method considered in \cite{Fulop2025} (choosing the coefficients $\alpha_b=\alpha_p=1$ for such method), which was shown to outperform the classical Euler method.
    
    In Figure \ref{fig:position_van_der_Waals_forced} it can be seen that for the same time step, the integration of the discrete thermodynamic forced Euler-Lagrange equations outperforms the GENERIC-motivated algorithm, although a non linear equation must be solved at each step, which increases the computational demand.

    A comparison when $p_A=0$ between these two methods is also shown in Figure \ref{fig:position_van_der_Waals_unforced_GENERIC}.


\begin{figure}[H]
    \centering
    \begin{subfigure}{.495\linewidth}
        \centering
        \includegraphics[width=\linewidth]{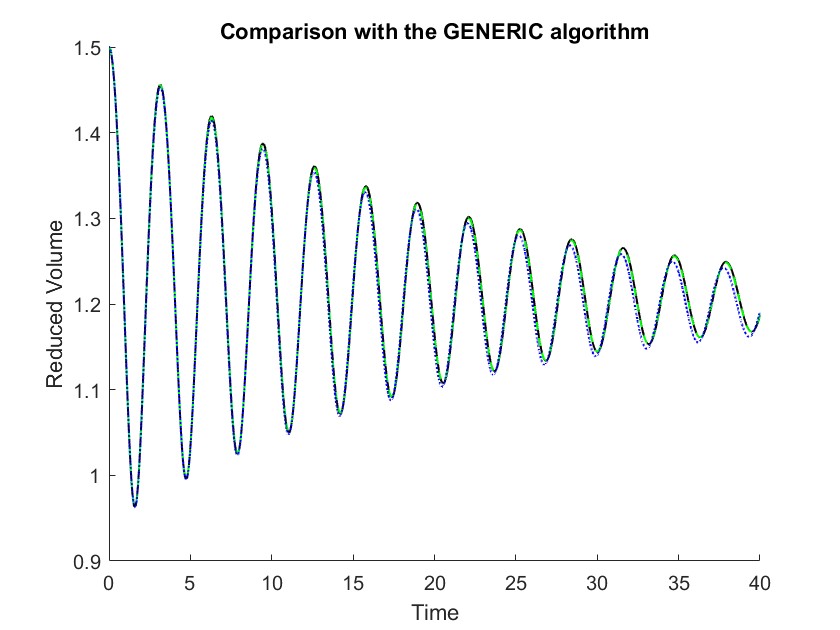}
        \caption{Reduced Volume}
    \end{subfigure}
    \begin{subfigure}{.495\linewidth}
        \centering
        \includegraphics[width=\linewidth]{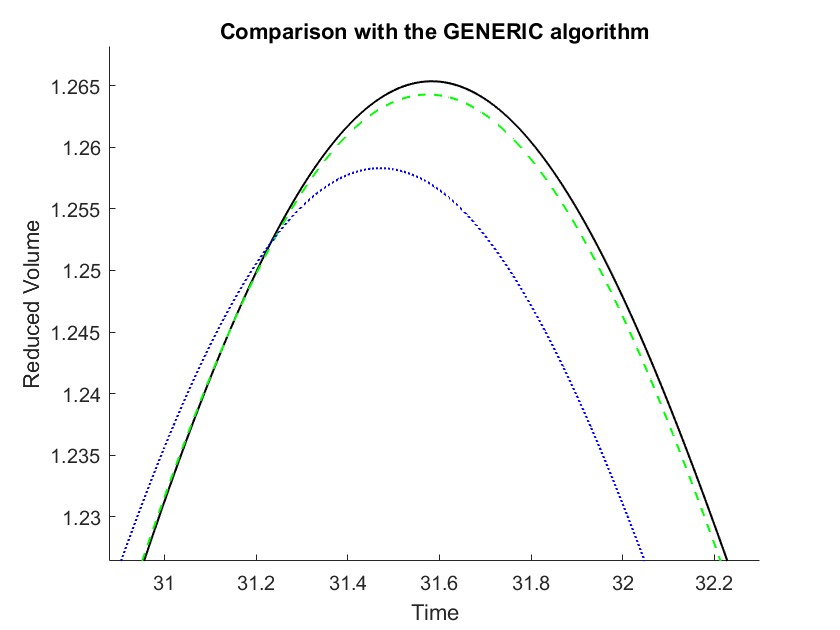}
        \caption{Detailed section}
    \end{subfigure}
    \caption{Integration of $V_r$ for a Van der Waals gas composed of monoatomic molecules contained in a cylinder with a time step $h=0.01$ using the Discrete thermodynamic Euler--Lagrange equations solution (in green and dashed) and the GENERIC-motivated symplectic method (in blue and dotted), compared with the solution given by a RK-4 method with a tolerance of $10^{-15}$ (in black and solid).}
    \label{fig:position_van_der_Waals_forced}
\end{figure}

\begin{figure}[H]
    \centering
    \begin{subfigure}{.495\linewidth}
        \centering
        \includegraphics[width=\linewidth]{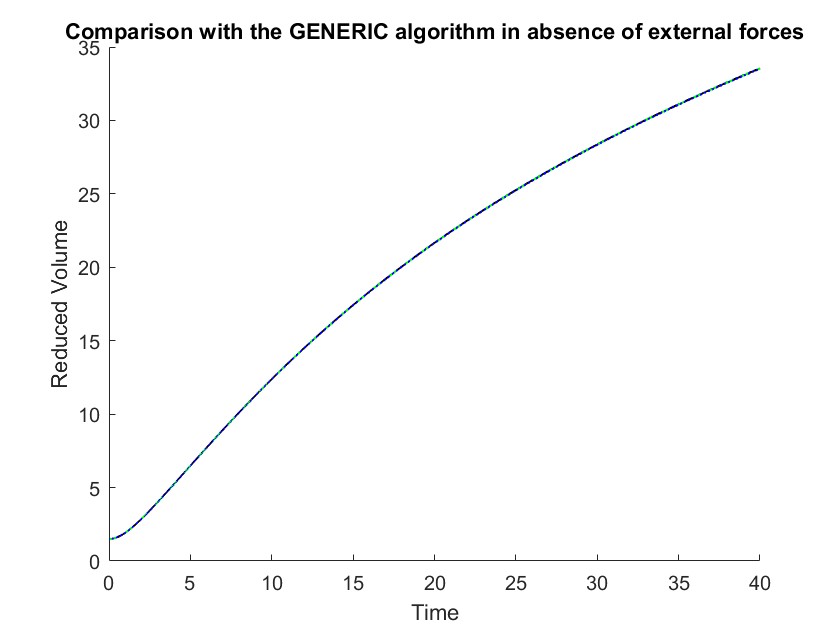}
        \caption{Reduced Volume}
    \end{subfigure}
    \begin{subfigure}{.495\linewidth}
        \centering
        \includegraphics[width=\linewidth]{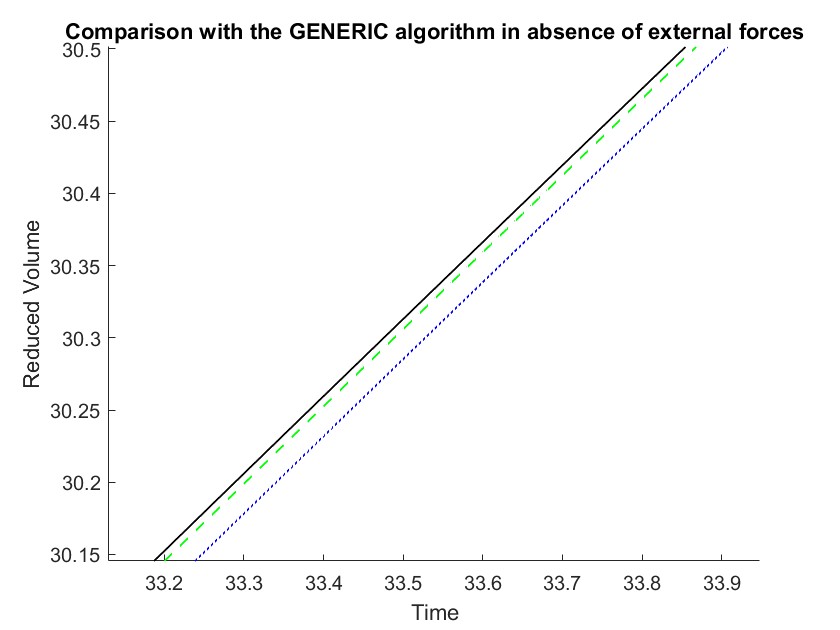}
        \caption{Detailed section}
    \end{subfigure}
    \caption{Integration of $V_r$ for a Van der Waals gas composed of monoatomic molecules contained in a cylinder in absence of external forces with a time step $h=0.01$ using the Discrete thermodynamic Euler--Lagrange equations solution (in green and dashed) and the GENERIC-motivated symplectic method (in blue and dotted), compared with the solution given by a RK-4 method with a tolerance of $10^{-15}$ (in black and solid).}
    \label{fig:position_van_der_Waals_unforced_GENERIC}
\end{figure}



\section{Conclusions and outlook}\label{sec:conclusions}

Inspired by the geometric framework of almost cosymplectic manifolds derived in \cite{d.B2025}, we have developed a discrete variational principle for adiabatically closed simple thermodynamic systems, leading to the discrete thermodynamic forced Euler--Lagrange equations \eqref{eq:EulerLagrangeDiscretas}. We have also presented a Noether's theorem for these systems (Theorem~\ref{theorem:Noether_discrete}) as well as for the continuous counterpart (Theorem~\ref{NoetherI} and Theorem~\ref{proposition:Cartan_symmetries_Lagrangian}. We have illustrated the effectiveness of our method with four examples: the damped harmonic oscillator, both a perfect gas and a Van der Waals gas confined by a piston contained in a cylinder, and a Van der Waals gas subject to external forces.

\bigskip

\noindent In future work, we shall pursue the following objectives:

\begin{itemize}

\item To extend these methods to the remaining thermodynamic systems considered in \cite{G.Y2018,Y.G2022}, namely systems with internal mass transfer, adiabatically closed non-simple thermodynamic systems, and open simple thermodynamic systems.

\item To investigate the relationship between thermodynamic systems and nonholonomic constraints (see \cite{lucia}). Furthermore, we would like to relate our approach to the symplectic one in \cite{ghosh2024hamiltonianthermodynamicssymplecticmanifolds}.

\item To delve into the geometry of almost cosymplectic manifolds, studying the analogues of (co)isotropic and Lagrangian submanifolds for these structures. An almost cosymplectic structure induces a bivector field, which in turn defines a bracket on the space of smooth functions. This bracket is not Poisson, except in trivial cases, since it does not satisfy the Jacobi identity. We shall study the properties of this bracket.

\end{itemize}


\section*{Acknowledgements}
The authors wish to express their gratitude to the anonymous referees for their constructive comments. 
M.~de León and J.~Bajo received financial support from Grant PID2022-137909NB-C21 funded by MICIU/AEI/ 10.13039/501100011033 and Grant CEX2023-001347-S funded by MICIU/AEI/10.13039/501100011033. J.~Bajo received the support of a fellowship from ``La Caixa'' Foundation (ID 100010434) with code LCF/BQ/PFA25/11000023.

\section*{Data availability}
The Matlab codes employed for the simulations are available upon reasonable request.

\section*{Conﬂict of Interest}
The authors declare no conﬂict of interest.

\let\emph\oldemph


\addcontentsline{toc}{section}{References}
\printbibliography

\end{document}